\documentclass[final,5p,times,twocolumn]{elsarticle}

\usepackage{amssymb}
\usepackage{amsmath}
\usepackage{subcaption}
\usepackage{bm}
\usepackage[draft]{hyperref}
\usepackage{booktabs}
\usepackage{xcolor}
\usepackage{amsthm}
\newtheorem{lemma}{Lemma}

\journal{Physica D: Nonlinear Phenomena}

\begin{document}

\thispagestyle{empty}
\onecolumn
\noindent {\large \\~\\~\\~\\~\\~\\~\\~\\ \copyright ~2020. This manuscript version is made available under the CC-BY-NC-ND 4.0 license \url{http://creativecommons.org/licenses/by-nc-nd/4.0/} \\~\\

\noindent This is a peer-reviewed and accepted version of the following in press document: \\~\\

\noindent K. Y. Ng and M. M. Gui. ``COVID-19: Development of a Robust Mathematical Model and Simulation Package with Consideration for Ageing Population and Time Delay for Control Action and Resusceptibility,'' {\it Physica D: Nonlinear Phenomena (In Press)}, 2020.}
\newpage

\begin{frontmatter}



  \title{COVID-19: Development of a Robust Mathematical Model and Simulation Package with Consideration for Ageing Population and Time Delay for Control Action and Resusceptibility}


  \author[a,c]{Kok Yew Ng\corref{cor1}}
  \ead{mark.ng@ulster.ac.uk}
  \author[b]{Meei Mei Gui}
  \ead{m.gui@qub.ac.uk}

  \cortext[cor1]{Corresponding author.}
  \address[a]{Nanotechnology and Integrated BioEngineering Centre (NIBEC), Ulster University, Jordanstown Campus, Shore Road, Newtownabbey BT37 0QB, UK.}
  \address[b]{School of Chemistry and Chemical Engineering, David Keir Building, Queen's University Belfast, Stranmillis Road, Belfast BT9 5AG, UK.}
  \address[c]{Electrical and Computer Systems Engineering, School of Engineering, Monash University, Malaysia.}

  \begin{abstract}
    The current global health emergency triggered by the pandemic COVID-19 is one of the greatest challenges mankind face in this generation. Computational simulations have played an important role to predict the development of the current pandemic. Such simulations enable early indications on the future projections of the pandemic and is useful to estimate the efficiency of control action in the battle against the SARS-CoV-2 virus. The SEIR model is a well-known method used in computational simulations of infectious viral diseases and it has been widely used to model other epidemics such as Ebola, SARS, MERS, and influenza A. This paper presents a modified SEIRS model with additional exit conditions in the form of death rates and resusceptibility, where we can tune the exit conditions in the model to extend prediction on the current projections of the pandemic into three possible outcomes; death, recovery, and recovery with a possibility of resusceptibility. The model also considers specific information such as ageing factor of the population, time delay on the development of the pandemic due to control action measures, as well as resusceptibility with temporal immune response. Owing to huge variations in clinical symptoms exhibited by COVID-19, the proposed model aims to reflect better on the current scenario and case data reported, such that the spread of the disease and the efficiency of the control action taken can be better understood. The model is verified using two case studies for verification and prediction studies, based on the real-world data in South Korea and Northern Ireland, respectively.
  \end{abstract}

\end{frontmatter}


\section{Introduction}\label{Intro}
The coronavirus disease COVID-19 is a respiratory infection disease caused by the novel coronavirus, SARS-CoV-2. The first COVID-19 outbreak was reported in Wuhan of Hubei Province, China at the end of December 2019. Within just two months, the disease has rapidly spread across the world and it has been declared a global pandemic in early March 2020. As of 20th April 2020, the virus has affected close to 2.5 million people with approximately 170,000 confirmed deaths across at least 184 countries \cite{LancetDong}.

The symptoms caused by the SARS-CoV-2 virus have large variations with most people only experiencing mild to moderate respiratory illnesses and only a smaller group of people would develop complications of respiratory failure or acute respiratory distress syndrome. Based on clinical data reported from Wuhan where the outbreak began, elderly patients have been identified to have higher odds to experience severe symptoms with higher mortality rate compared to people of younger age \cite{Zhou2020}. Study also shows that up to approximately 80\% of the people infected with SARS-CoV-2 are asymptomatic  carriers, i.e. they experience no or mild symptoms but are still able to transport the virus to others \cite{prompet2020}. This has caused the detection and containment of the SARS-CoV-2 virus to be much more challenging. As a result, social distancing has been widely implemented in many countries worldwide to slow down the transportation of the virus through minimised human-to-human contact. Individuals who have recovered from COVID-19 after experiencing mild or moderate symptoms are more likely to develop temporary resistant towards the virus and are unlikely to experience severe respiratory illnesses \cite{Thevarajan2020}. However, in rare occasions, there have been clinical findings showing that patients who have recovered from the disease have been tested positive again. For instance, in February 2020, a patient in Osaka, Japan, has been tested positive towards the SARS-CoV-2 a few days after being discharged from the hospital for treatment with the disease \cite{Jap2020}. Due to very limited knowledge on the immune response of humans to this novel virus, the possibility of reinfection cannot be ruled out at the moment.

Mathematical modelling and computational simulations have played important roles in describing the dynamics of infectious diseases using nonlinear systems so that their risks could be better understood \cite{vanden2017,PENG2013,ZHANG2012,Keeling,vanden2008,Lekone2006}. Most notably, the SEIR (Susceptible-Exposed-Infected-Removed/Recovered) model has been widely reported during the past decades in quantitative modelling studies of infectious epidemic/pandemic, such as the severe acute respiratory syndrome (SARS) in 2002 \cite{Dye2003}, the influenza A (pH1N1) pandemic in 2009 \cite{Saito2013}, the Middle East respiratory syndrome (MERS) pandemic in 2013 \cite{Chowell2014}, as well as the latest Ebola outbreak in 2018 \cite{Diaz2018}. The SEIR model represents a typical infectious epidemic disease using four distinctive phases; susceptible (S) represents the population that has yet to be infected by the virus; exposed (E) represents the number of individuals exposed to the virus; infected (I) models the number of people infected, have demonstrated symptoms, and are able to spread the virus to the people in the S compartment; and lastly, recovered/removed (R) models the number of people who have recovered and are assumed to have immune response towards the virus \cite{LIN2020, Lopez2020, Kucharski2020, LaSen}. Thus, based on the model, the S compartment will slowly deplete as the outbreak prolongs further, and the virus will eventually die out due to insufficient population within the S compartment. This compartmental modelling method allows transport of population from one compartment to another, where the disease transmission rates with respect to time can be simulated.

In this work, we propose a modified SEIRS model based on the Kermack-McKendrick model \cite{kermack1927} with consideration for time delay and resusceptibility to the virus after recovery. In this model, the probability of a recovered patient to be reinfected with SARS-CoV-2 is taken into consideration to predict the future projection of COVID-19 cases. Resusceptibility is one of the crucial keys that could possibility lead to a prolonged COVID-19 pandemic. Time delay in the control action representing the time taken for the authorities to act on the virus and also the duration of short-term immunity after recovery, which may lead to resusceptibility, are also considered in the model. The time delay factor is applied to enhance the robustness and accuracy of the model and simulation, hence to better reflect on the timely situation with specific measures to control the transmission of the disease. The consideration for resusceptibility with time delay is an important highlight in this paper as it has rarely been considered in the SEIR models reported thus far \cite{Kucharski2020, Prem2020, Fang2020}. Other than that, we also included information such as demographic details for the ageing population, who seem to experience a higher fatality rate due to COVID-19 \cite{covid}.

This paper is organised as follows: Section \ref{Model} introduces the mathematical model; Section \ref{Stability} presents the theoretical proofs for the positivity, boundedness, and stability of the model, as well as describes the model with time delay factors for the control action and potential resusceptibility; Section \ref{Case} verifies the proposed model and also provides some extensive results and discussions on predictions using the model through case studies based on real-world data; and Section \ref{Conclusion} concludes the paper. \ref{Simulation} presents the design of the simulation package using the MATLAB/Simulink environment.

\section{Modelling COVID-19 Using Modified SEIRS}\label{Model}
First, let's consider the modified SEIRS model system below,
\begin{align}
  \frac{dS(t)}{dt} &= \Lambda - \mu S(t) - \beta (1 - \sigma)S(t)I(t) + R_s(t), \label{eq:dSnotau} \\
  \frac{dE(t)}{dt} &= \beta (1 - \sigma)S(t)I(t) - (\mu + \alpha)E(t), \label{eq:dEnotau} \\
  \frac{dI(t)}{dt} &= \alpha E(t) - (\mu + \gamma)I(t) - \Delta I(t) \label{eq:dInotau} \\
  \frac{dR(t)}{dt} &= \gamma I(t) - \mu R(t) - R_s(t), \label{eq:dRnotau} \\
  \frac{dD(t)}{dt} &= \Delta I(t), \label{eq:dDnotau}
\end{align}
where $S(t), E(t), I(t), R(t),$ and $D(t)$ represent the susceptible, exposed, infected, recovered/removed, and deaths compartments, respectively. It is established that $S(t) + E(t) + I(t) + R(t) + D(t) = N(t)$, where $N(t)$ is the total stock population. The constant $\Lambda$ is the birth rate in the overall population and $\mu$ is the death rate due to conditions other than COVID-19. The parameter $\beta$ is the rate of transmission per S-I contact, $\alpha$ is the rate of which an exposed person becomes infected, and $\gamma$ is the recovery rate. Therefore, the incubation and recovery times are $\tau_{inc} = \frac{1}{\alpha}$ and $\tau_{rec} = \frac{1}{\gamma}$, respectively. The constant $\sigma$ is the efficiency of the control action to reduce the infection rate and to flatten the curve. It has a direct effect on the basic reproduction number $R_0$, which will be explained further later in this paper.

The parameter $\Delta = \delta \left[(1 - \kappa_{old})N_{old} + (1 - \kappa)(1 - N_{old}) \right]$ comes into effect in the worst case scenario where the patient does not recover from the virus. We model the fatality rate with influence of the fraction of elderly population (above 65 years of age) within the community $N_{old}$, where the percentages of nonelderly and elderly who recovered are $\kappa$ and $\kappa_{old}$, respectively. The time spent hospitalised or infected in fatal cases is $\tau_{hosp} = \frac{1}{\delta}$. In this paper, we establish that $\tau_{hosp} = \tau_{rec}$, assuming that patients spend the same amount of time hospitalised or infected, whether they recover from the virus or not.

The function $R_s(t)$ represents the resusceptible stock, which can be computed from the recovered population using
\begin{equation}
  R_s(t) = \xi R(t), \label{eq:Rs1}
\end{equation}
where $\xi$ is the percentage of the recovered population who are resusceptible to the virus. The number of actual recovered/removed cases with permanent immunity can then be written using
\begin{equation}
  R_c(t) = R(t) - R_s(t). \label{eq:Rc1}
\end{equation}
In an ideal situation, population who recovered develop permanent immunity against the virus, i.e. $\xi = 0$. As a result, (\ref{eq:Rs1}) becomes $R_s(t) = 0$ and subsequently, (\ref{eq:Rc1}) becomes $R_c(t) = R(t)$.

\section{Positivity, Boundedness, and Equilibrium Analysis of the Model}\label{Stability}
\subsection{Positivity of the Solutions}
\begin{lemma}\label{lemma:pos}
  The solutions to all subpopulations $(S(t),E(t),I(t),R(t),D(t))$ in the system (\ref{eq:dSnotau})--(\ref{eq:dDnotau}) are nonnegative for all time $t \geq 0$ given any finite nonnegative initial conditions of $(S(0) \geq 0, ~E(0) \geq 0, ~I(0) \geq 0, ~R(0) \geq 0, ~D(0) \geq 0)$.
\end{lemma}
\begin{proof}
  Firstly, it is established that all subpopulations $(S(t),E(t),I(t),R(t),D(t))$ defined by the system (\ref{eq:dSnotau})--(\ref{eq:dDnotau}) are continuously differentiable. As such, if all subpopulations have nonnegative initial conditions, and that if any of the subpopulations is zero at time $t = t_i \geq 0$, its derivative is nonnegative by inspection. Assume that $S(0) \geq 0$, $S(t_1) = 0$, and $R_s(t_1) \geq 0$ at time instant $t = t_1$. Then, we can rewrite (\ref{eq:dSnotau}) using
  $$\frac{dS(t_1)}{dt} = \Lambda + R_s(t_1) \geq 0,$$
  where we can establish that $S(t_1^+) \geq 0$ and hence, $S(t)$ is nonnegative for all time $t \geq 0$. Next, assume that $E(0) \geq 0$, $E(t_2) = 0, S(t_2) \geq 0,$ and $I(t_2) \geq 0$ at time instant $t = t_2$. We can rewrite (\ref{eq:dEnotau}) using
  $$\frac{dE(t_2)}{dt} = \beta(1 - \sigma)S(t_2)I(t_2) \geq 0,$$
  so that $E(t_2^+) \geq 0$ and hence, $E(t)$ is nonnegative for all time $t \geq 0$. Assume that $I(0) \geq 0$, $I(t_3) = 0$, and $E(t_3) \geq 0$ at time instant $t = t_3$. Equation (\ref{eq:dInotau}) then becomes
  $$\frac{dI(t_3)}{dt} = \alpha E(t_3) \geq 0,$$
  where we can establish that $I(t_3^+) \geq 0$ and hence, $I(t)$ is nonnegative for all time $t \geq 0$. Assume that $R(0) \geq 0$, $R(t_4) = 0$, and $I(t_4) \geq 0$ at time instant $t = t_4$. Equation (\ref{eq:dRnotau}) can be rewritten using
  $$\frac{dR(t_4)}{dt} = \gamma I(t_4) \geq 0,$$
  so that $R(t_4^+) \geq 0$ and hence, $R(t)$ is nonnegative for all time $t \geq 0$. Finally, assume that $D(0) \geq 0$, $D(t_5) = 0$, and $I(t_5) \geq 0$ at time instant $t = t_5$. We can then rewrite (\ref{eq:dDnotau}) using
  $$\frac{dD(t_5)}{dt} = \Delta I(t_5) \geq 0,$$
  where we can establish that $D(t_5^+) \geq 0$ and hence, $D(t)$ is nonnegative for all time $t \geq 0$.

  It can be seen that since none of the subpopulations would have a negative derivative at any time instant of $t = t_i$ when all other subpopulations are nonnegative, then it can be concluded that all subpopulations are nonnegative for all time $t \geq 0$. As a result, given that $N(t) = S(t) + E(t) + I(t) + R(t) + D(t)$, the stock population $N(t)$ is also nonnegative for all time $t \geq 0$. Hence, the proof is complete.
\end{proof}

\subsection{Boundedness of the Solutions}
\begin{lemma}
  The stock population $N(t)$ is finitely upperbounded for any nonnegative initial conditions.
\end{lemma}
\begin{proof}
  The dynamics of the stock population can be written using
  \begin{align}
    \frac{dN(t)}{dt} &= \frac{dS(t)}{dt} + \frac{dE(t)}{dt} + \frac{dI(t)}{dt} + \frac{dR(t)}{dt} + \frac{dD(t)}{dt}, \notag \\
    &= \Lambda - \mu(S(t) + E(t) + I(t) + R(t)), \notag \\
    &= \Lambda - \mu(N(t) - D(t)). \label{eq:dNdt1}
  \end{align}
  Assuming that $N(t) >> D(t)$, and since Lemma \ref{lemma:pos} has established that all subpopulations are nonnegative and given that all parameters are assumed to be positive, then (\ref{eq:dNdt1}) becomes
  \begin{equation}
    \frac{dN(t)}{dt} \approx \Lambda -\mu N(t).
  \end{equation}
  We can then deduce that
  \begin{equation}
    \frac{dN(t)}{dt} \leq \Lambda - \mu N(t), \label{eq:dNdt}
  \end{equation}
  where an integration of the inequality (\ref{eq:dNdt}) yields
  $$N(t) \leq N(0)e^{-\mu t} + \frac{\Lambda}{\mu}\left(1 - e^{-\mu t} \right) \leq max\left(N(0), \frac{\Lambda}{\mu}\right),$$
  for all time $t \geq 0$. As a result, the stock population is finitely upperbounded and hence, the proof is complete.
\end{proof}

\subsection{Disease-free Equilibrium}\label{section:DFE}
\begin{lemma}\label{lemma1}
  The disease-free equilibrium $E_{DFE}$ is locally asymptotically stable if the basic reproduction number $R_0 < 1$.
\end{lemma}
\begin{proof}
  The disease-free equilibrium can be obtained by equating equations (\ref{eq:dSnotau})–(\ref{eq:dRnotau}) to zero, hence satisfying
  \begin{align}
    \Lambda - \mu S(t) - \beta (1 - \sigma)S(t)I(t) + R_s(t) &= 0, \\
    \beta (1 - \sigma)S(t)I(t) - (\mu + \alpha)E(t) &= 0, \\
    \alpha E(t) - (\mu + \gamma)I(t) - \Delta I(t) &= 0, \\
    \gamma I(t) - \mu R(t) - R_s(t) &= 0,
  \end{align}
  of which the disease-free equilibrium is given by $E_{DFE} = \left(\frac{\Lambda}{\mu}, 0, 0, 0 \right)$. Equation (\ref{eq:dDnotau}) can be removed from this analysis without loss of generality as other equations do not depend on it. It can then be shown that the Jacobian for (\ref{eq:dSnotau})–(\ref{eq:dRnotau}) at $E_{DFE}$ is written using
  \begin{equation}
    J_{DFE} = \left[ \begin{array}{cccc}
    -\mu & 0 & -\frac{\beta \Lambda (1 - \sigma)}{\mu} & \xi \\
    0 & -(\mu + \alpha) & \frac{\beta \Lambda (1 - \sigma)}{\mu} & 0 \\
    0 & \alpha & -(\mu + \gamma + \Delta) & 0 \\
    0 & 0 & \gamma & -(\mu + \xi)
    \end{array} \right].
  \end{equation}

  The characteristic equation can subsequently be obtained by subtracting $\lambda$ from the diagonal elements and then computing the determinant, which yields
  \begin{equation}
    (-\mu -\lambda)(-(\mu+\xi) -\lambda)f_1(\lambda) = 0, \label{eq:det1}
  \end{equation}
  where
  \begin{equation}
    f_1(\lambda) = (-(\mu+\gamma+\Delta) -\lambda)(-(\mu+\alpha) -\lambda) -\frac{\alpha \beta \Lambda (1 -\sigma)}{\mu}. \label{eq:f1}
  \end{equation}

  The first two eigenvalues in (\ref{eq:det1}) can be easily computed to be $\lambda_1 = -\mu, ~\lambda_2 = -\mu -\xi$, and that they are negative. As for the remaining eigenvalues, they can be found by solving the quadratic $f_1(\lambda)$ in (\ref{eq:f1}), which can be expended and represented using
  \begin{equation}
    a_1\lambda^2 + a_2\lambda + a_3 = 0, \notag
  \end{equation}
  where
  \begin{align}
    a_1 &= 1, \notag \\
    a_2 &= 2\mu + \gamma + \alpha + \Delta, \notag \\
    a_3 &= (\mu + \gamma + \Delta)(\mu + \alpha) - \frac{\alpha \beta \Lambda(1-\sigma)}{\mu}. \notag
  \end{align}

  For the disease-free equilibrium to be stable, i.e. all eigenvalues are negative, it is required that
  \begin{align}
    \frac{\alpha \beta \Lambda(1-\sigma)}{\mu} - (\mu + \gamma + \Delta)(\mu + \alpha) &< 0, \notag \\
    \frac{\alpha \beta \Lambda(1-\sigma)}{\mu (\mu + \gamma + \Delta)(\mu + \alpha)} &< 1 \notag.
  \end{align}

  As such, the basic reproduction number with the control action is defined using
  \begin{equation}
    R_0 = \frac{\alpha \beta \Lambda (1 - \sigma)}{\mu(\mu + \alpha)(\mu + \gamma + \Delta)}, \label{eq:R01}
  \end{equation}
  where for a disease-free system that is locally asymptotically stable, we need to ensure that $R_0 < 1$ while an unstable $E_{DFE}$ would translate to $R_0 > 1$. Hence, the proof is complete.
\end{proof}

Should there be no control action taken, i.e. $\sigma = 0$, then the basic reproduction number in (\ref{eq:R01}) becomes
\begin{equation}
  R_0 = \frac{\alpha \beta \Lambda}{\mu(\mu + \alpha)(\mu + \gamma + \Delta)}, \nonumber
\end{equation}
which is similar to other models found in the literature \cite{vanden2017,Keeling,vanden2008}. See \cite{Stewart2020} and the references within for a brief study on using control theory to fight COVID-19.

\subsection{Endemic Equilibrium}
\begin{lemma}\label{lemma2}
  The endemic equilibrium $E_{EE}$ is locally asymptotically stable if the basic reproduction number $R_0 > 1$.
\end{lemma}
\begin{proof}
  Let's assume that the system (\ref{eq:dSnotau})–(\ref{eq:dRnotau}), which other than the disease-free equilibrium has a unique equilibrium at $E_{EE} = (S^*, E^*, I^*, R^*)$ such that
  \begin{align}
    S^* &= \frac{\Lambda}{\mu R_0}, \notag \\
    E^* &= \frac{\mu \Lambda(\mu + \gamma + \Delta)(\mu + \xi)(R_0 - 1)}{\alpha \beta \Lambda (1 - \sigma)(\mu + \xi) - \alpha \mu \gamma \xi R_0}, \notag \\
    I^* &= \frac{\mu \Lambda(\mu + \xi)(R_0 - 1)}{\beta
    \Lambda (1 - \sigma)(\mu + \xi) - \mu \gamma \xi R_0}, \notag \\
    R^* &= \frac{\mu \gamma \Lambda (R_0 - 1)}{\beta
    \Lambda (1 - \sigma)(\mu + \xi) - \mu \gamma \xi R_0}. \notag
  \end{align}

  It can be seen that the model has no positive endemic equilibrium if $R_0 < 1$ for $E^*, I^*,$ and $R^*$ would be negative, which indicate an unrealistic biological system. If $R_0 = 1$, then we would have the disease-free equilibrium $E_{DFE}$ discussed earlier in Section \ref{section:DFE}. Hence, for a positive endemic equilibrium system, we would require that $R_0 > 1$.

  For the stability analysis of the endemic equilibrium, we use the Jacobian for (\ref{eq:dSnotau})–(\ref{eq:dRnotau}) at $E_{EE}$, which is written using
  \begin{equation}
    \hspace{-2mm}J_{EE} = \left[\hspace{-2mm} \begin{array}{cccc}
    -(\mu +\kappa I^*) & 0 & -\kappa S^* & \xi \\
    \kappa I^* & -(\mu + \alpha) & \kappa S^* & 0 \\
    0 & \alpha & -(\mu + \gamma + \Delta) & 0 \\
    0 & 0 & \gamma & -(\mu + \xi)
    \end{array} \hspace{-2mm}\right],
  \end{equation}
  where $\kappa = \beta (1 - \sigma)$. The characteristic equation can subsequently be obtained by subtracting $\lambda$ from the diagonal elements and then computing the determinant, which yields
  \begin{equation}
    (-(\mu+\xi)-\lambda)(-\lambda^3 + b_1 \lambda^2 - b_2 \lambda + b_3) - \alpha \gamma \kappa \xi I^* = 0, \label{eq:chareq2}
  \end{equation}
  where
  \begin{align}
    b_1 &= -(\mu + \kappa I^*) -(\mu + \alpha) -(\mu + \gamma + \Delta), \notag \\
    &= -3\mu -\alpha -\gamma -\Delta -\kappa I^*, \notag \\
    \begin{split}
      b_2 &= (\mu + \kappa I^*)(\mu + \alpha) + (\mu + \kappa I^*)(\mu + \gamma + \Delta) \\
      &\qquad + (\mu + \alpha)(\mu + \gamma + \Delta)- \alpha \kappa S^*,
      \end{split} \notag \\
      \begin{split}
        &= 3\mu^2 + 2\mu \alpha + 2\mu \gamma  + 2 \mu \Delta + \alpha \gamma + \gamma \Delta  \\
        &\qquad + \kappa I^* (2 \mu + \alpha + \gamma + \Delta) - \alpha \kappa S^*,
        \end{split} \notag \\
        b_3 &= -(\mu + \kappa I^*)((\mu + \alpha)(\mu + \gamma + \Delta) - \alpha \kappa S^*) - \alpha \kappa^2 S^* I^*, \notag \\
        &= -(\mu +\kappa I^*)(\mu^2 + \mu(\alpha + \gamma + \Delta) + \alpha (\gamma + \Delta)) + \mu \alpha \kappa S^*. \notag
      \end{align}
      Assuming that ~~$\alpha \gamma \kappa \xi I^* << (-(\mu+\xi)-\lambda)(-\lambda^3 + b_1 \lambda^2 - b_2 \lambda + b_3),$
      then the characteristic equation in (\ref{eq:chareq2}) becomes
      \begin{equation}
        (-(\mu+\xi)-\lambda)(-\lambda^3 + b_1 \lambda^2 - b_2 \lambda + b_3) \approx 0.
      \end{equation}

      It is obvious that the first eigenvalue is $\lambda_1 = -\mu -\xi$. It can also be shown that in order for the remaining eigenvalues to be negative such that the endemic equilibrium is locally asymptotically stable, i.e. $b_1 < 0, ~b_2 > 0,$ and $b_3 < 0,$ we would require that $R_0 > 1$ while an unstable $E_{EE}$ would translate to $R_0 < 1$. Hence, the proof is complete.
    \end{proof}

    For the remaining of this paper where we verify the model and perform predictions using real-world data in Section \ref{Case}, and also for the model used in the simulation package presented in \ref{Simulation}, we assume a closed population with negligible birth and death rates, i.e. $\frac{\Lambda}{\mu} \approx 1, ~\Lambda \approx 0, ~\mu \approx 0$. Time delay factors are also considered for the control action taken as well as resusceptibility. As a result, the system (\ref{eq:dSnotau})–(\ref{eq:dDnotau}) becomes
    \begin{align}
      \frac{dS(t)}{dt} &= - \beta (I(t)S(t) - \sigma I(t-\tau_\sigma)S(t-\tau_\sigma)) + R_s(t,\tau_\xi), \label{eq:dS} \\
      \frac{dE(t)}{dt} &= \beta (I(t)S(t) - \sigma I(t-\tau_\sigma)S(t-\tau_\sigma)) - \alpha E(t), \label{eq:dE} \\
      \frac{dI(t)}{dt} &= \alpha E(t) - \gamma I(t) - \Delta I(t), \label{eq:dI} \\
      \frac{dR(t)}{dt} &= \gamma I(t) - R_s(t,\tau_\xi), \label{eq:dR} \\
      \frac{dD(t)}{dt} &= \Delta I(t). \label{eq:dD}
    \end{align}
    and hence, the basic reproduction number in (\ref{eq:R01}) becomes
    \begin{equation}
      R_0 = \frac{\beta(1 - \sigma)}{\gamma + \Delta}. \label{eq:R02}
    \end{equation}

    The time delay $\tau_\sigma = \tau_{pre-\sigma}+\tau_{post-\sigma}$ indicates the time taken for the control action to take effect in flattening the infection curve, where $\tau_{pre-\sigma} \geq 0$ represents the time to initiate the control action after the first confirmed case at $t=0$, and $\tau_{post-\sigma} \geq 0$ represents the time after the control action has been initiated but before the effects are evidenced in the outputs of the system. In practical scenarios, $\tau_{post-\sigma}$ can be used to model the delay for the population to effectively respond to the rules introduced by the control action, such as social distancing, self-isolation, and lockdown.

    The function $R_s(t,\tau_\xi)$ represents the resusceptible stock with the consideration for temporal immunity, of which we can then rewrite (\ref{eq:Rs1}) using
    \begin{equation}
      R_s(t,\tau_\xi) = \xi R(t - \tau_\xi), \label{eq:Rs}
    \end{equation}
    where the time delay $\tau_\xi \geq 0$ represents the duration of temporal immune response of the recovered population. Hence, we can also update the expression for the number of recovered cases introduced in (\ref{eq:Rc1}) using
    \begin{equation}
      R_c(t) = R(t) - R_s(t,\tau_\xi). \label{eq:Rc}
    \end{equation}
    In an ideal situation where population who recovered develop permanent immunity against the virus, $\xi = 0$ and $\tau_\xi \rightarrow \infty$. As a result, (\ref{eq:Rs}) becomes $R_s(t,\infty) = 0$ and subsequently, (\ref{eq:Rc}) can be rewritten as $R_c(t) = R(t)$.

    The block diagram of the proposed SEIRS model with time delay is shown in Figure \ref{fig:BD}. The system with time delay is assumed to be stable and will exhibit similar disease-free equilibrium and endemic equilibrium properties as the system without time delay provided that the time delay parameters are nonnegative, i.e. $\tau_{\sigma} \geq 0, ~\tau_\xi \geq 0$. Detailed discussions on the theoretical stability analysis of SEIR and similar epidemic models with time delay can be found in studies such as \cite{LaSen, Tipsri2015, Huang2010, yan_liu_2006}.
    \begin{figure}
      \centering
      \includegraphics[width=\columnwidth]{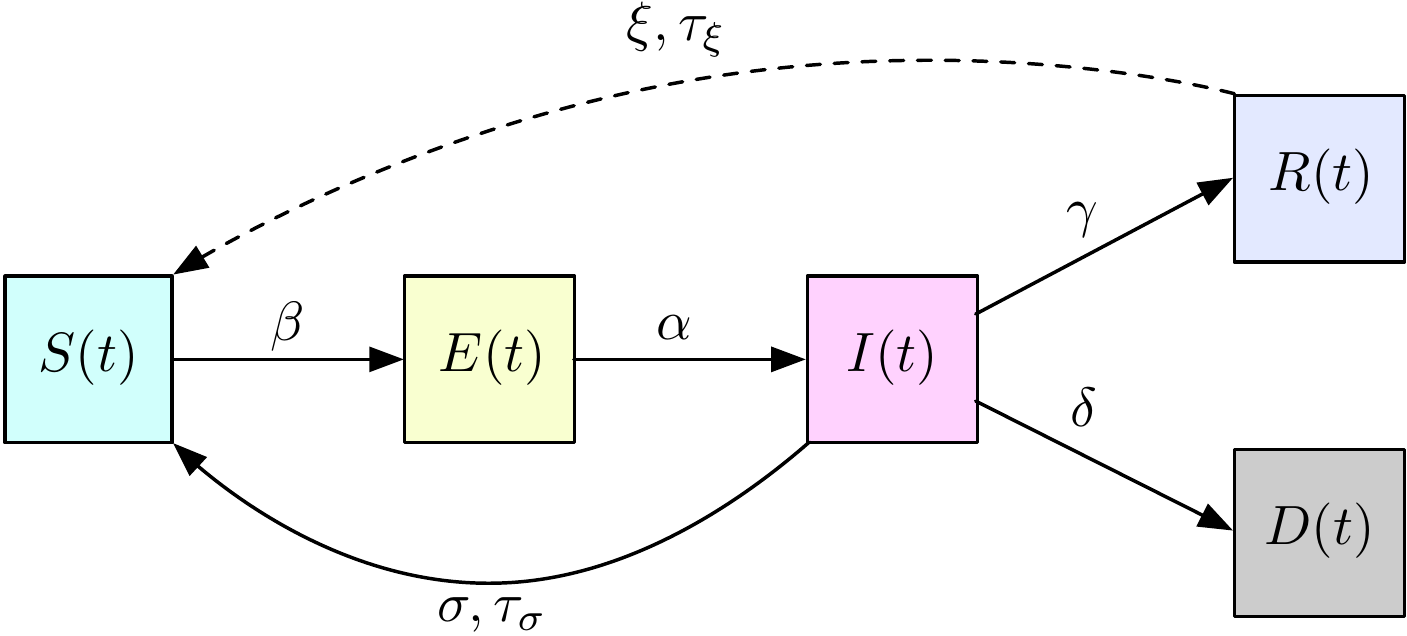}
      \caption{The block diagram of the proposed SEIRS system used in the simulation package.}
      \label{fig:BD}
    \end{figure}

    \section{Case Studies}\label{Case}
    \subsection{Case Study 1: Verification using Data in South Korea}\label{Korea}
    South Korea is used as a case study due to the amount of data available given that it is one of the first few countries to be directly affected by the virus outside of China. Its first confirmed case was reported on 20th January 2020 \cite{SHIM2020}. The other reason is that South Korea is also one of the very few countries that managed to effectively flatten the infection curve and it has set itself apart from others in leading the fight against COVID-19, at least for the moment. For example, the country started vigorous testing among its population with contact tracing, especially those of confirmed and suspected cases during the early stage of the epidemic. The government accomplished this by maintaining a public database keeping track of mobile phones, credit cards, and other data of patients who tested positive \cite{normile2020}. Also, on 16th March 2020, the authorities in the country began screening every person, both domestic and international, who arrived at its airports.

    As of 20th April 2020, there have been 10,674 confirmed cases and 236 fatalities in South Korea \cite{LancetDong}. As a result, we used the following parameters for our simulation. First, we assumed that the population of South Korea to be approximately $N =$ 51.5$\times 10^6$ with an elderly population of about 15\% ($N_{old} =$ 0.15) \cite{oecd2018}. We then set the percentage of recovery to be 98\% ($\kappa =$ 0.98) for the general public \cite{LancetDong} and a fatality rate of $8\%$ ($(1 - \kappa_{old}) =$ 0.08) for the elderly \cite{KCDC2020}. We then assumed the incubation time and recovery time to be $\tau_{inc} =$ 5.1, and $\tau_{rec} =$ 18.8, respectively in accordance with \cite{flaxman2020}. The basic reproduction number was set to $R_0 =$ 5.1 (95\% CI: 4.97--5.23) based on the early growth-rate of the epidemic in South Korea. The initial infected and exposed cases were assumed to be $I(0) = 4$ and $E(0) = 20I(0)$, respectively. See Table \ref{tab:Korea}. Figure \ref{fig:KoreaNA1} shows the results of the initial fitting of the model based on the data in South Korea while Figure \ref{fig:KoreaNA2} shows the projections of the model when no control action is taken. There are some minor discrepancies between the modelled values and the real-world data during the initial stage of the simulation as seen in Figure \ref{fig:KoreaNA1}. This is absolutely reasonable and acceptable while modelling an actual epidemic as most countries are still coming to terms with the virus during the first month and the data do not usually represent the actual number of cases due to lack of testing for confirmed cases.
    \begin{table}
      \centering
      \caption{Initial parameters used to fit the model to the data in South Korea.}
      \label{tab:Korea}
      {\small
      \begin{tabular}{lc}
        \toprule
        {\bf Parameter}                           & {\bf Value} \\
        \midrule
        Stock population, $N$                     & 51.5$\times 10^6$  \\
        Fraction of elderly population, $N_{old}$ & 0.15                \\
        Percentage of recovery, $\kappa$          & 0.98                \\
        Fatality rate for elderly, $1 - \kappa_{old}$ & 0.08            \\
        Incubation time, $\tau_{inc}$             & 5.1 days            \\
        Recovery time, $\tau_{rec}$               & 18.8 days           \\
        Basic reproduction number, $R_0$          & 5.1 (95\% CI: 4.97--5.23) \\
        Initial infected cases, $I(0)$            & 4                   \\
        Initial exposed cases, $E(0)$             & 80                  \\
        \bottomrule
      \end{tabular}}
    \end{table}

    \begin{figure*}[t!]
      \centering
      \begin{subfigure}{\columnwidth}
        \centering
        \includegraphics[width=\columnwidth]{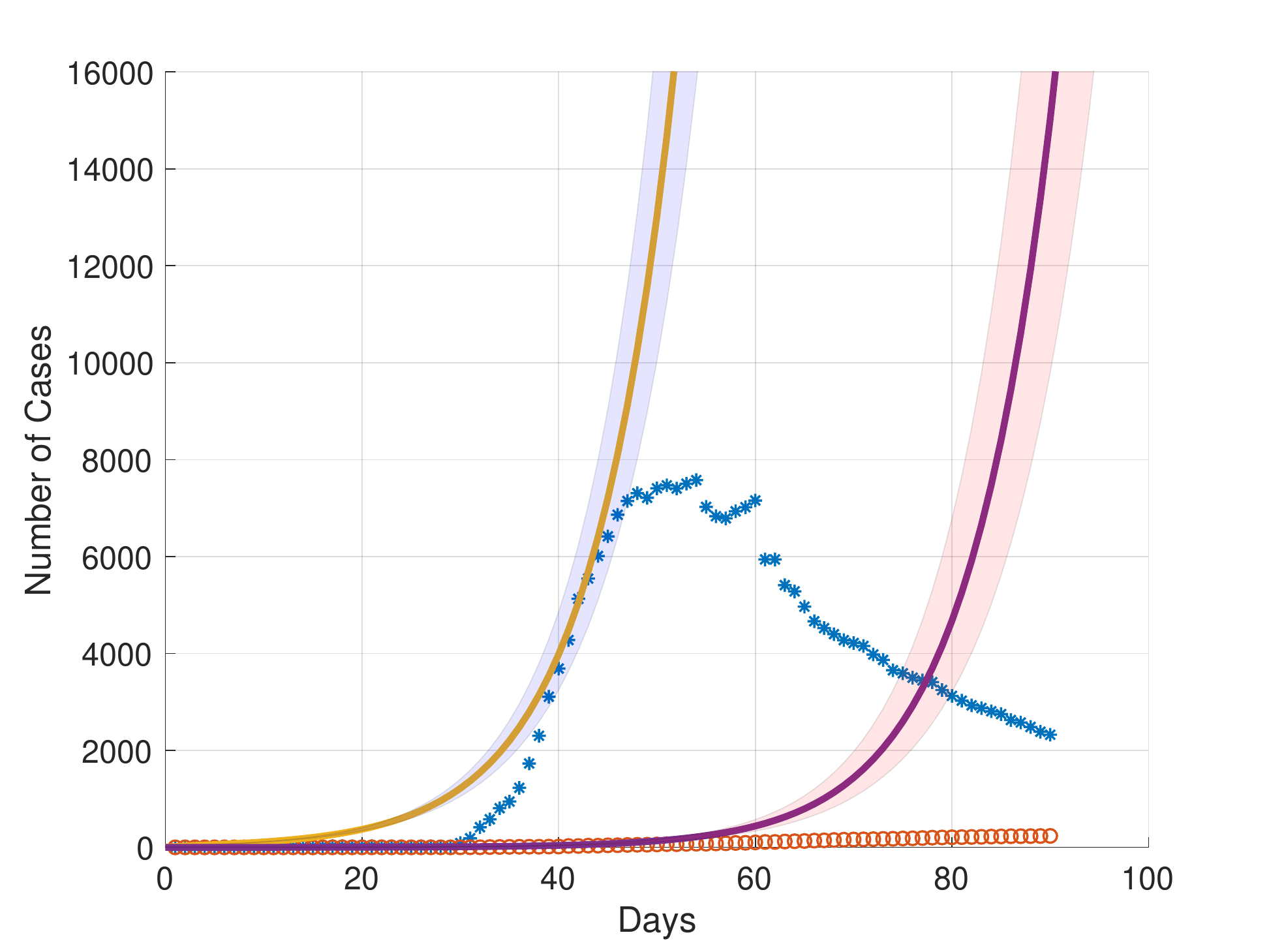}
        \caption{ }
        \label{fig:KoreaNA1}
      \end{subfigure}
      \begin{subfigure}{\columnwidth}
        \centering
        \includegraphics[width=\columnwidth]{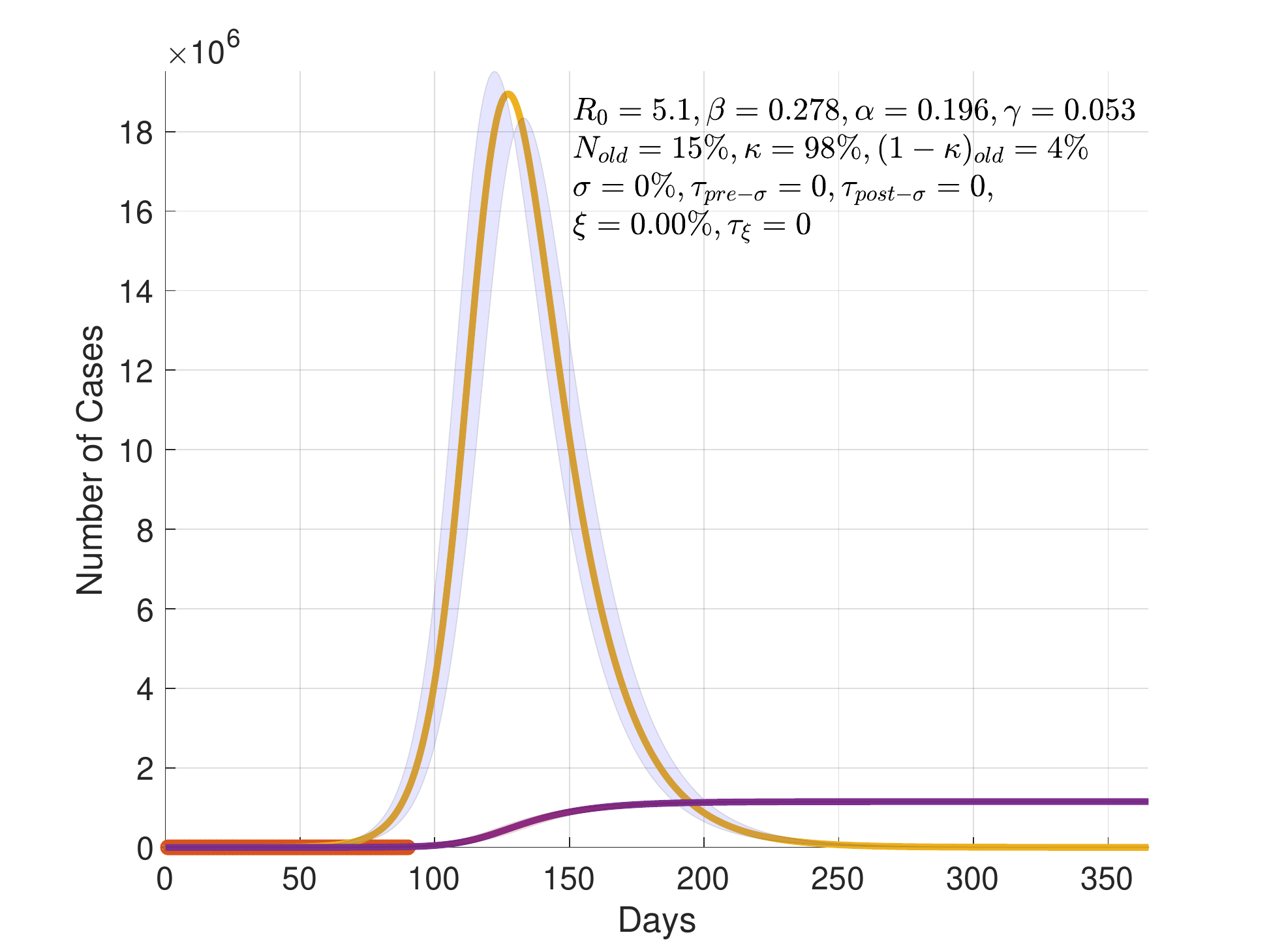}
        \caption{ }
        \label{fig:KoreaNA2}
        \end{subfigure} \\
        \begin{subfigure}{\textwidth}
          \centering
          \includegraphics[width=0.3\textwidth]{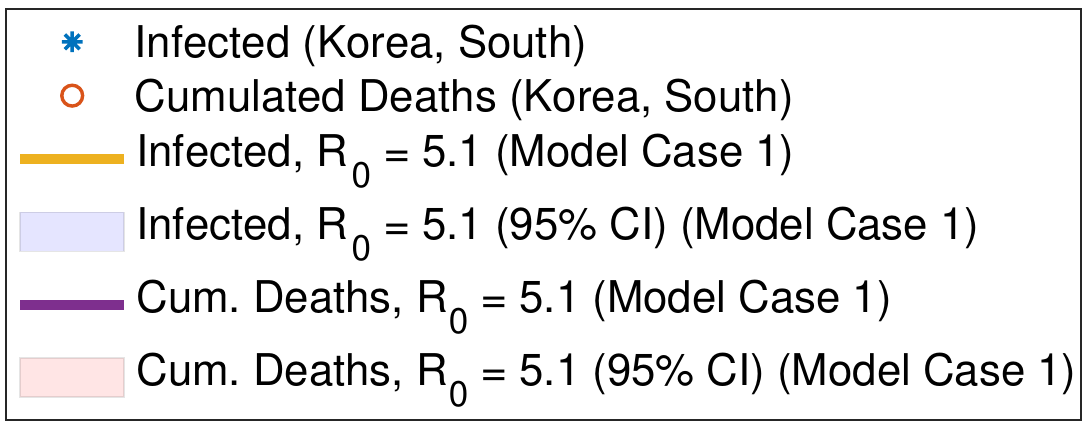}
        \end{subfigure}
        \caption{Subfigure (a) shows the initial fitting of the model onto the data in South Korea and subfigure (b) shows the projections of the model when no control action is taken.}
      \end{figure*}

      Once we have the initial fitting of the model, we introduced control action in line with the mitigation and preventive measures taken by the government. Due to the aforementioned vigorous testing, contact tracing, and isolation efforts taken, we assumed that the control action has an efficiency of 88\% ($\sigma =$ 0.88). As a result, the reproduction number could be reduced to $R_0 \approx$ 0.61. We also assumed that there was a time delay of 30 days since the first confirmed case before the control action was introduced ($\tau_{pre-\sigma} = 30$) and a further delay of approximately 13 days before the control action could be properly executed in the community ($\tau_{post-\sigma} = 13$). Figure \ref{fig:KoreaA} shows the simulation results. Figure \ref{fig:Korea1} shows that the trajectory of the modelled infected and death cases match the real-world data after the control action was introduced. Figure \ref{fig:Korea2} shows the simulation results until the model stabilises assuming no subsequent control action being taken to further reduce the reproduction number. Comparing Figure \ref{fig:Korea2} with Figure \ref{fig:KoreaNA2}, the peak of the number of infected cases could be reduced from about 19 million cases to about 7,500 cases.
      \begin{figure*}[t!]
        \centering
        \begin{subfigure}[t]{\columnwidth}
          \centering
          \includegraphics[width=\columnwidth]{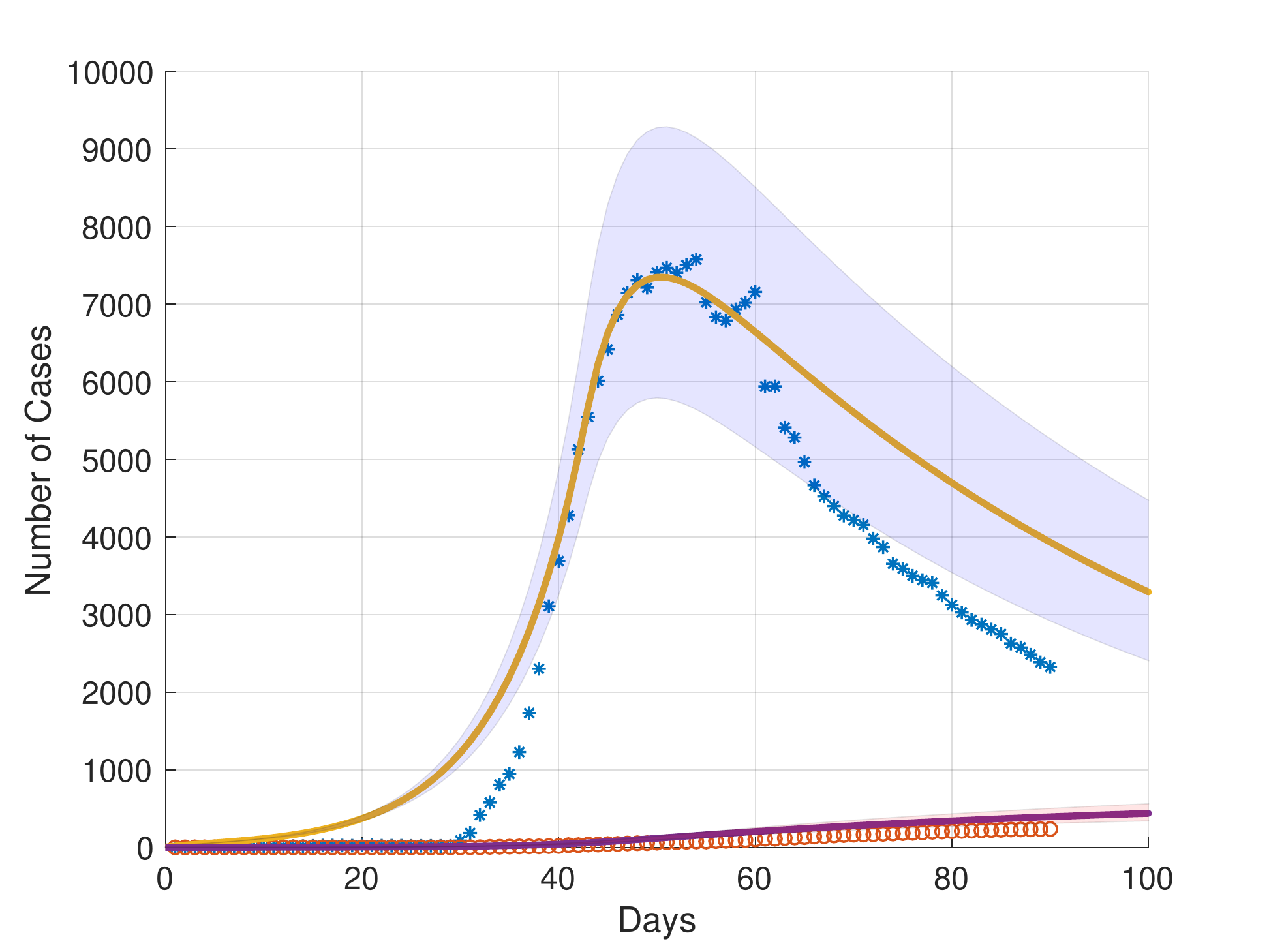}
          \caption{ }
          \label{fig:Korea1}
        \end{subfigure}
        \begin{subfigure}[t]{\columnwidth}
          \centering
          \includegraphics[width=\columnwidth]{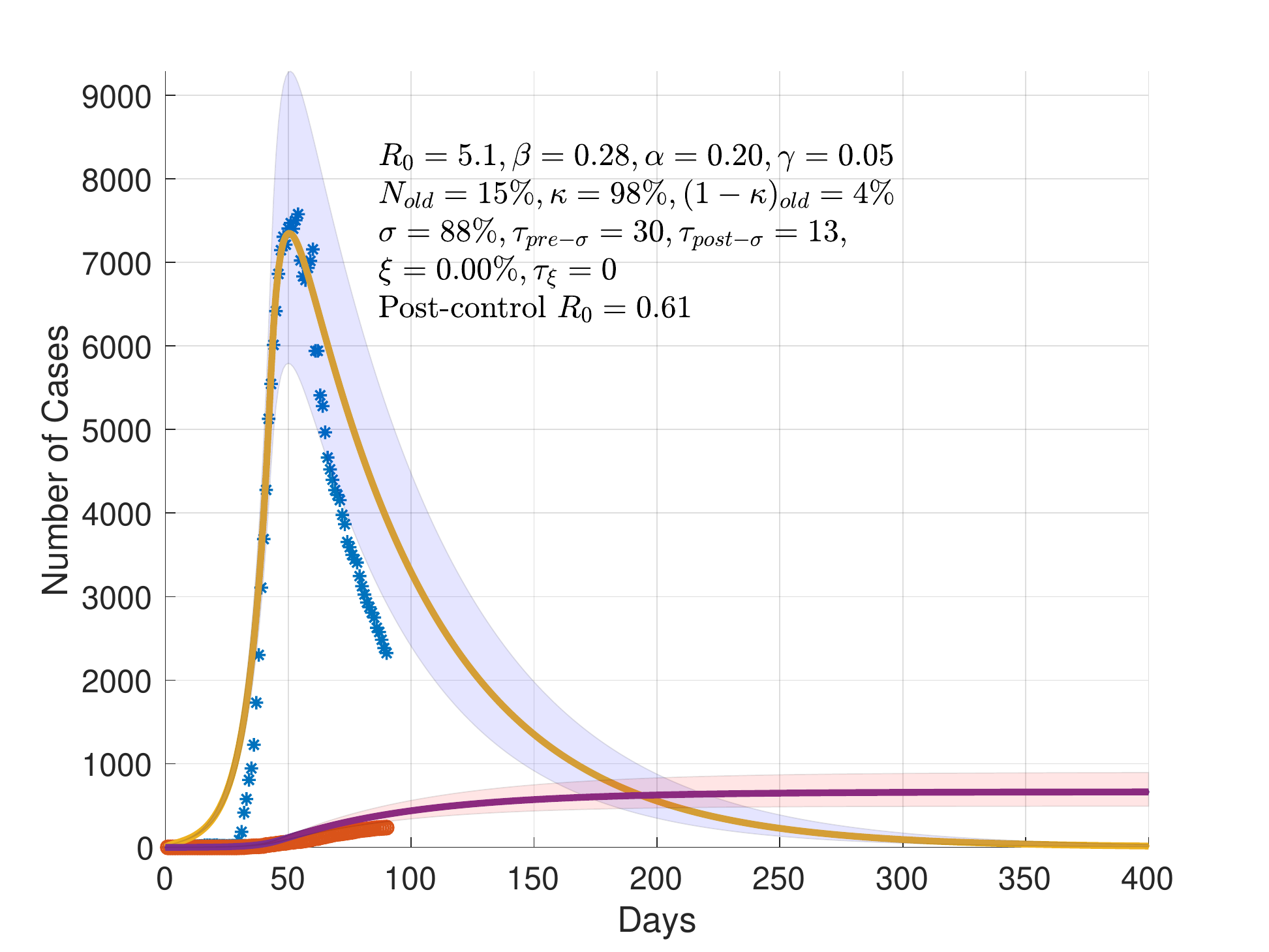}
          \caption{ }
          \label{fig:Korea2}
          \end{subfigure}  \\
          \begin{subfigure}{\textwidth}
            \centering
            \includegraphics[width=0.3\textwidth]{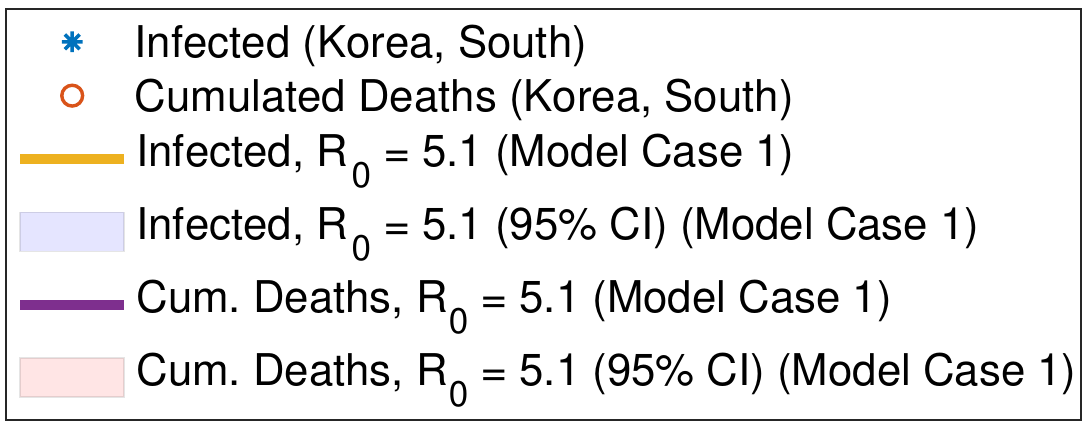}
          \end{subfigure}
          \caption{The projections of the model onto the data in South Korea when control action with an efficiency of 88\% is taken. Subfigure (a) shows the projections during the first 100 days while subfigure (b) shows the projections until the system achieves stability.}
          \label{fig:KoreaA}
        \end{figure*}

        Subsequently, beginning 20th March 2020, stronger infectious disease control measures for travellers coming from overseas were enforced, where all Koreans and foreigners with residence in Korea arriving from all countries would be subject to self-quarantine for 14 days upon entry. All short-term travellers will also be ordered to self-quarantine with exceptions for some limited special cases \cite{KCDC20202}. Around the same time, the Korea Centres for Disease Control and Prevention (KCDC) also started advising all people in the country to observe an ``Enhanced Social Distancing Campaign'' \cite{KoreaTimeline}.  Inducing these control actions into the model produces the results shown in Figure \ref{fig:KoreaB}. The second control action adds another efficiency of 50\%, hence bringing the reproduction number further down to $R_0 \approx 0.31$.
        \begin{figure*}[t!]
          \centering
          \begin{subfigure}{\columnwidth}
            \centering
            \includegraphics[width=\columnwidth]{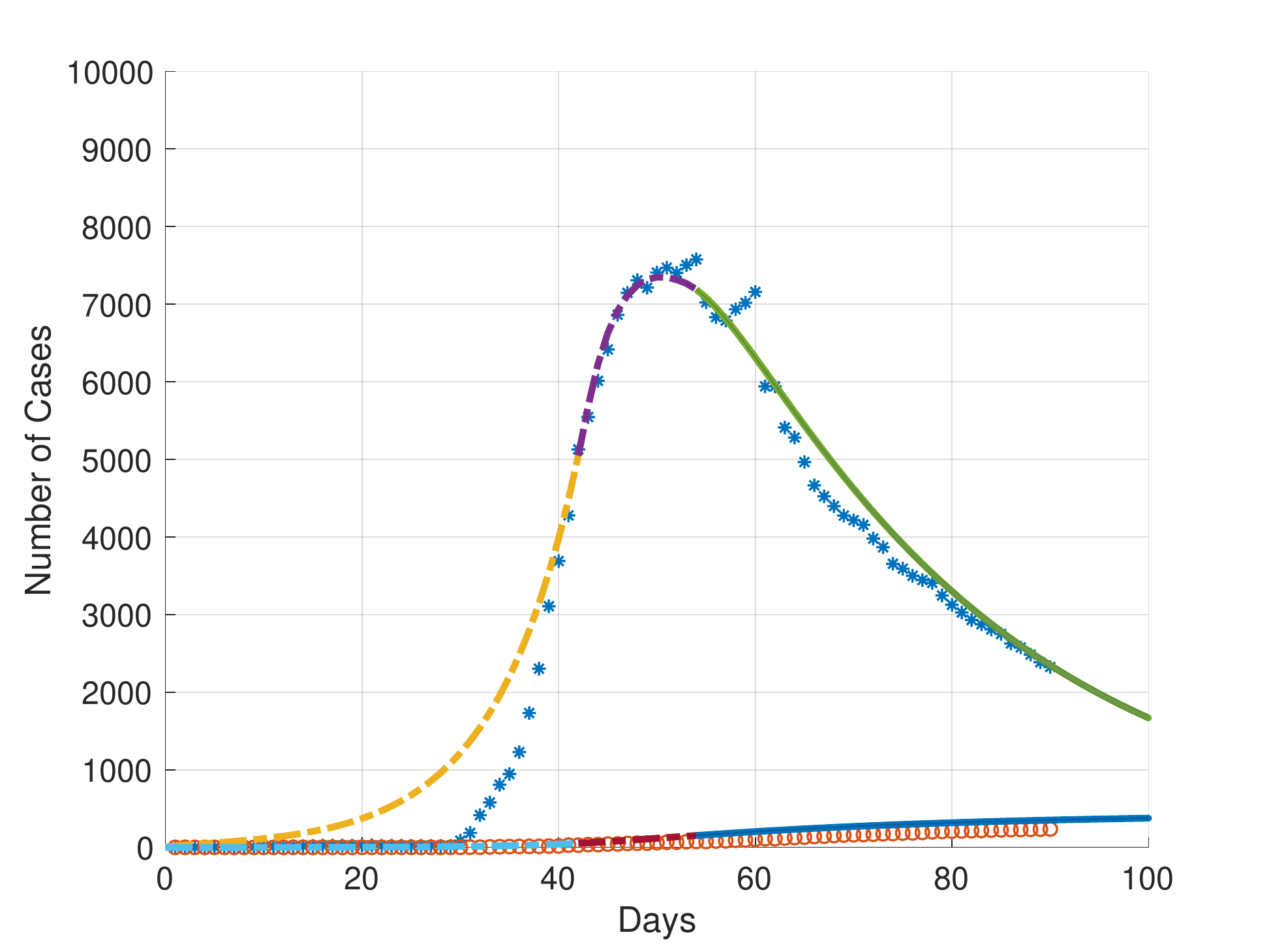}
            \caption{ }
            \label{fig:Korea1b}
          \end{subfigure}
          \begin{subfigure}{\columnwidth}
            \centering
            \includegraphics[width=\columnwidth]{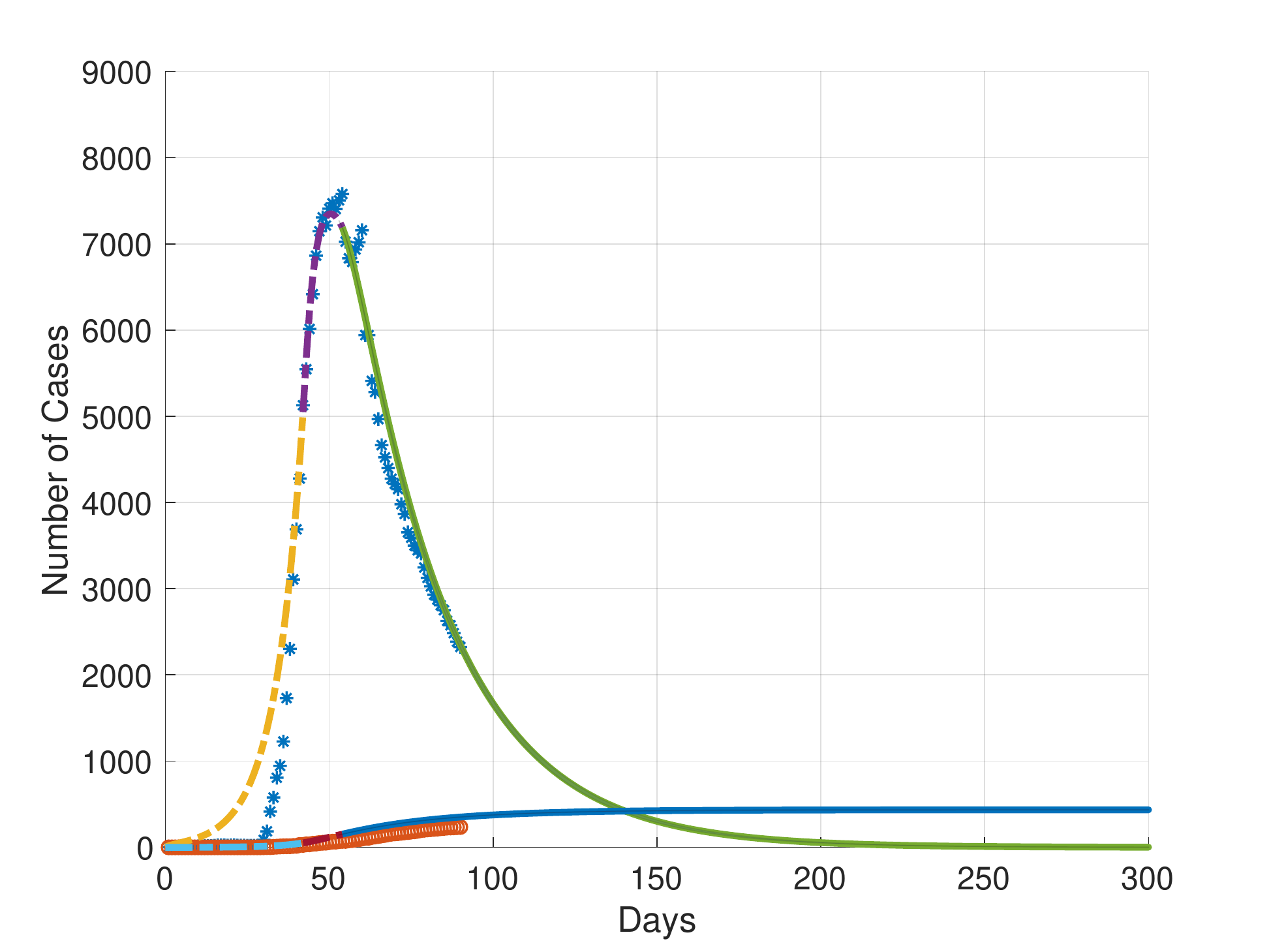}
            \caption{ }
            \label{fig:Korea2b}
            \end{subfigure} \\
            \begin{subfigure}{\textwidth}
              \centering
              \includegraphics[width=0.3\textwidth]{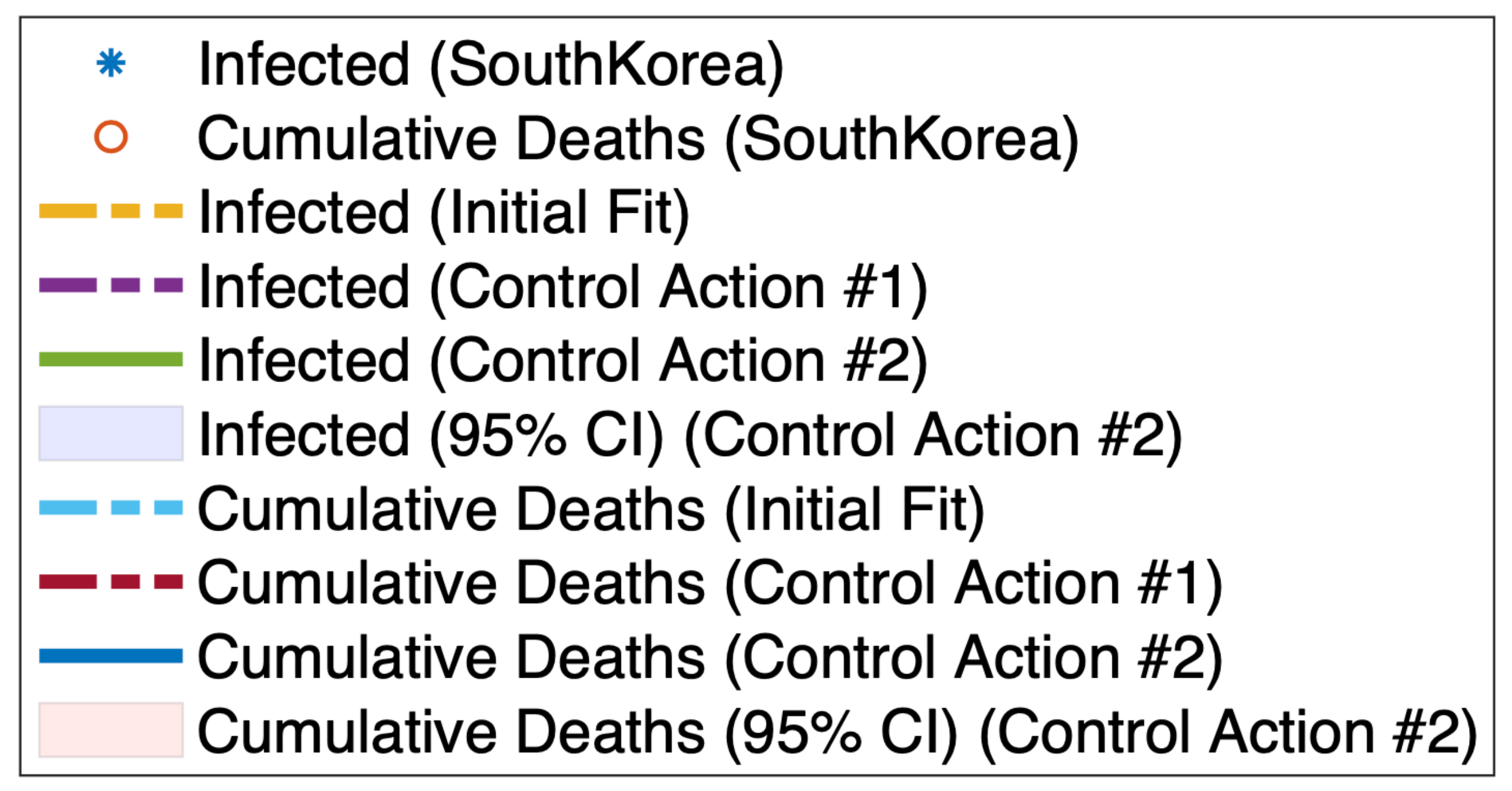}
            \end{subfigure}
            \caption{The projections of the model onto the data in South Korea when a second control action with an efficiency of 50\% is taken. Subfigure (a) shows the projections during the first 100 days while subfigure (b) shows the projections up till 300 days.}
            \label{fig:KoreaB}
          \end{figure*}

          \subsubsection{Simulation with Resusceptibility}
          One of the many uncertainties about COVID-19 is whether patients who have recovered from the virus will be reinfected in the future. There have been reports in the news that patients who recovered from the virus were tested positive for a second time after being cleared of the virus \cite{Jap2020, Kor2020, Dahl2020}. On the other hand, most health authorities believe that patients who recovered may develop an immunity towards the virus. However, it is not sure if the said immunity is short-term or long-term. Hence, further research is required to provide clinical proofs to this hypothesis.
          \begin{figure*}[t!]
            \begin{subfigure}{\columnwidth}
              \centering
              \includegraphics[width=\columnwidth]{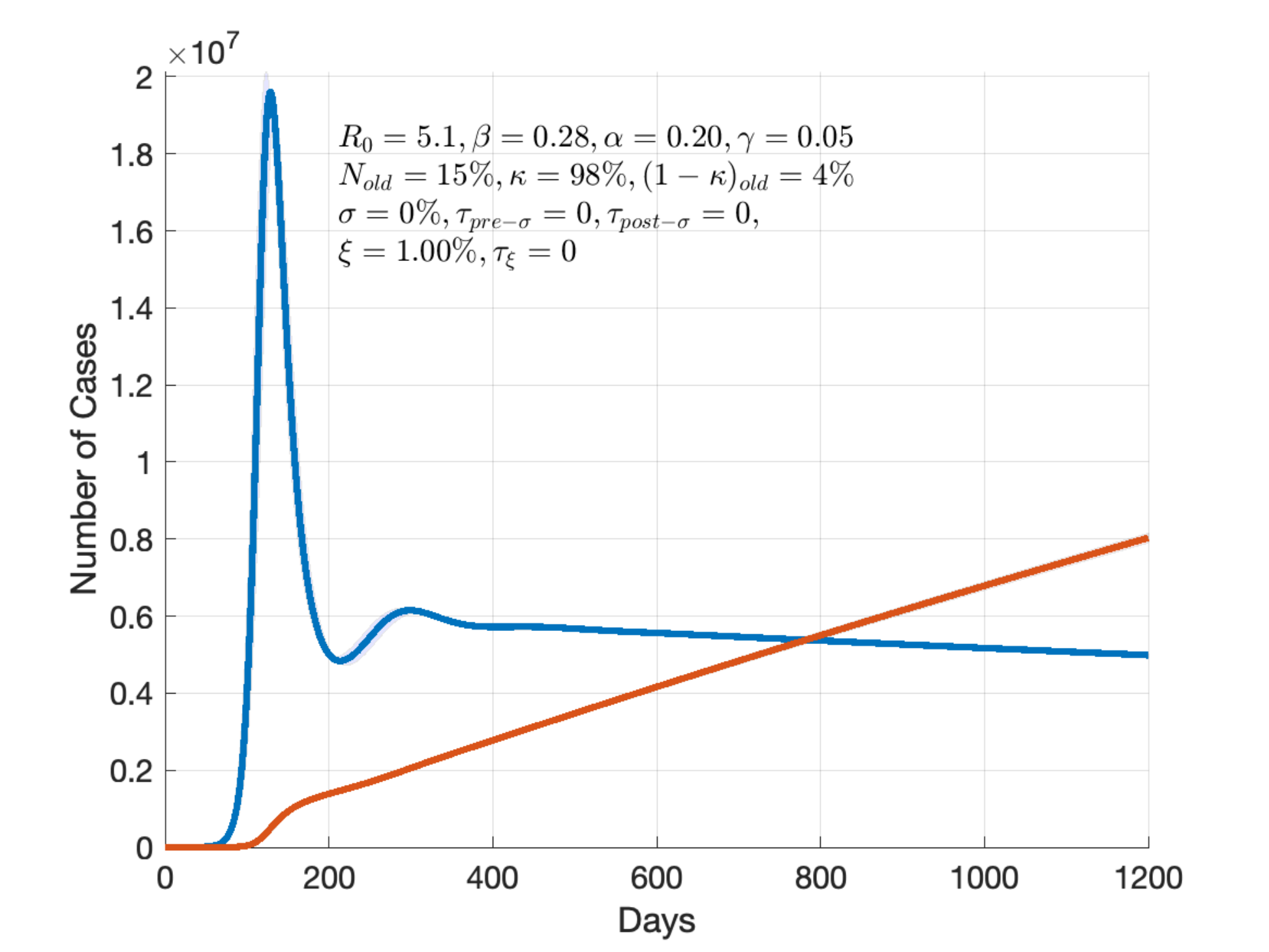}
              \caption{ }
              \label{fig:Korea4}
            \end{subfigure}
            \begin{subfigure}{\columnwidth}
              \centering
              \includegraphics[width=\columnwidth]{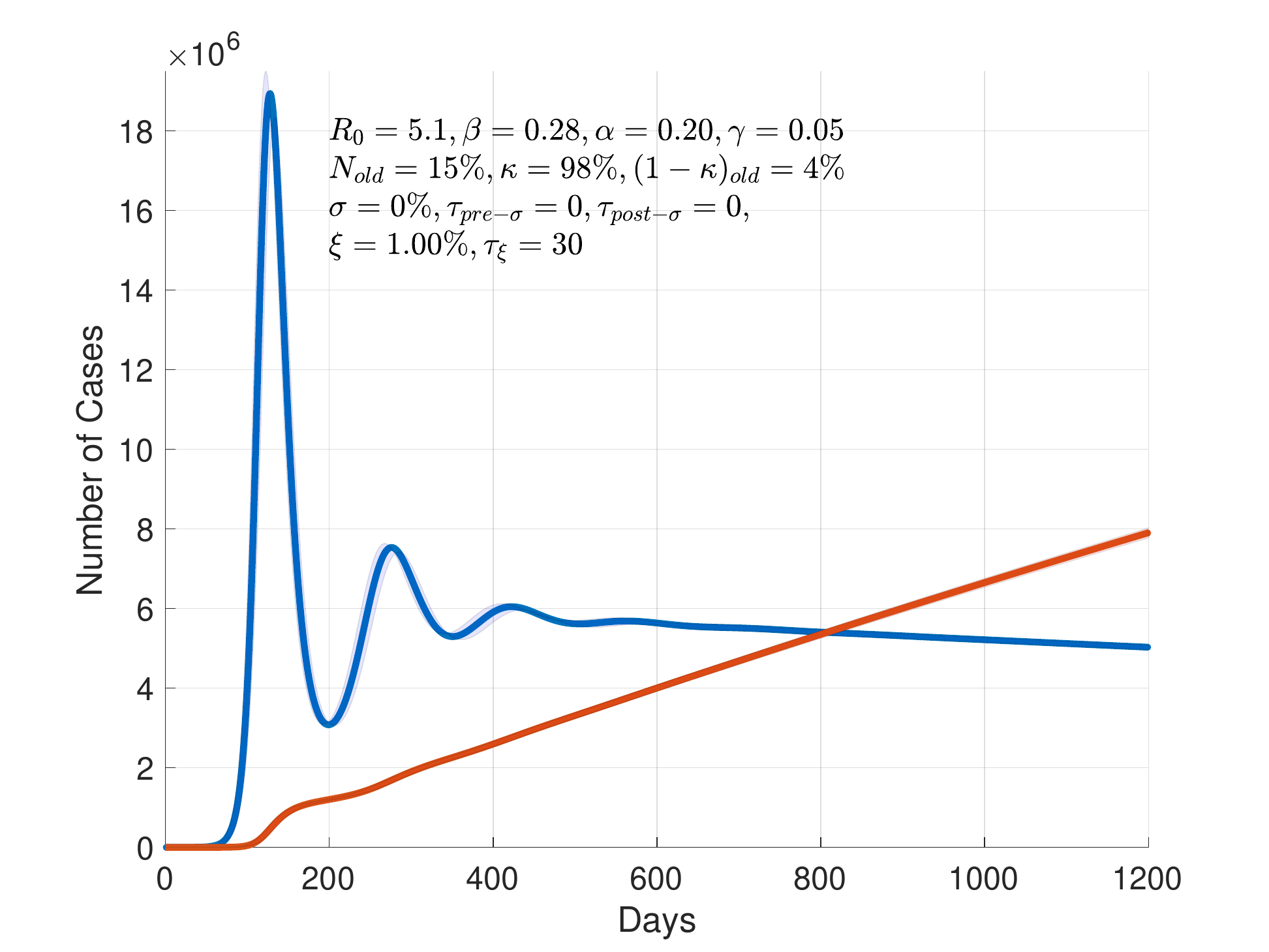}
              \caption{ }
              \label{fig:Korea5}
            \end{subfigure}
            \newline
            \begin{subfigure}{\columnwidth}
              \centering
              \includegraphics[width=\columnwidth]{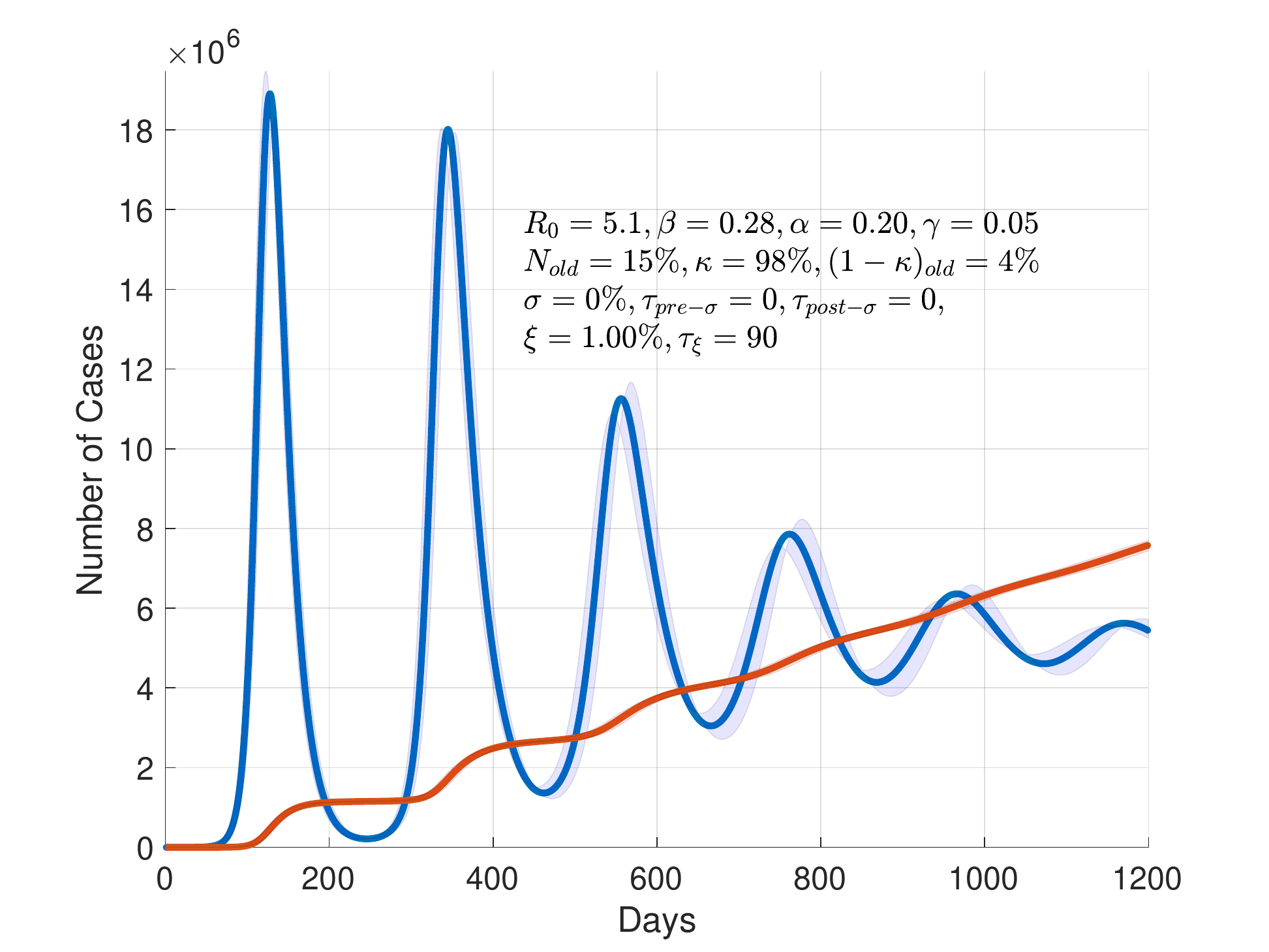}
              \caption{ }
              \label{fig:Korea6}
            \end{subfigure}
            \begin{subfigure}{\columnwidth}
              \centering
              \includegraphics[width=\columnwidth]{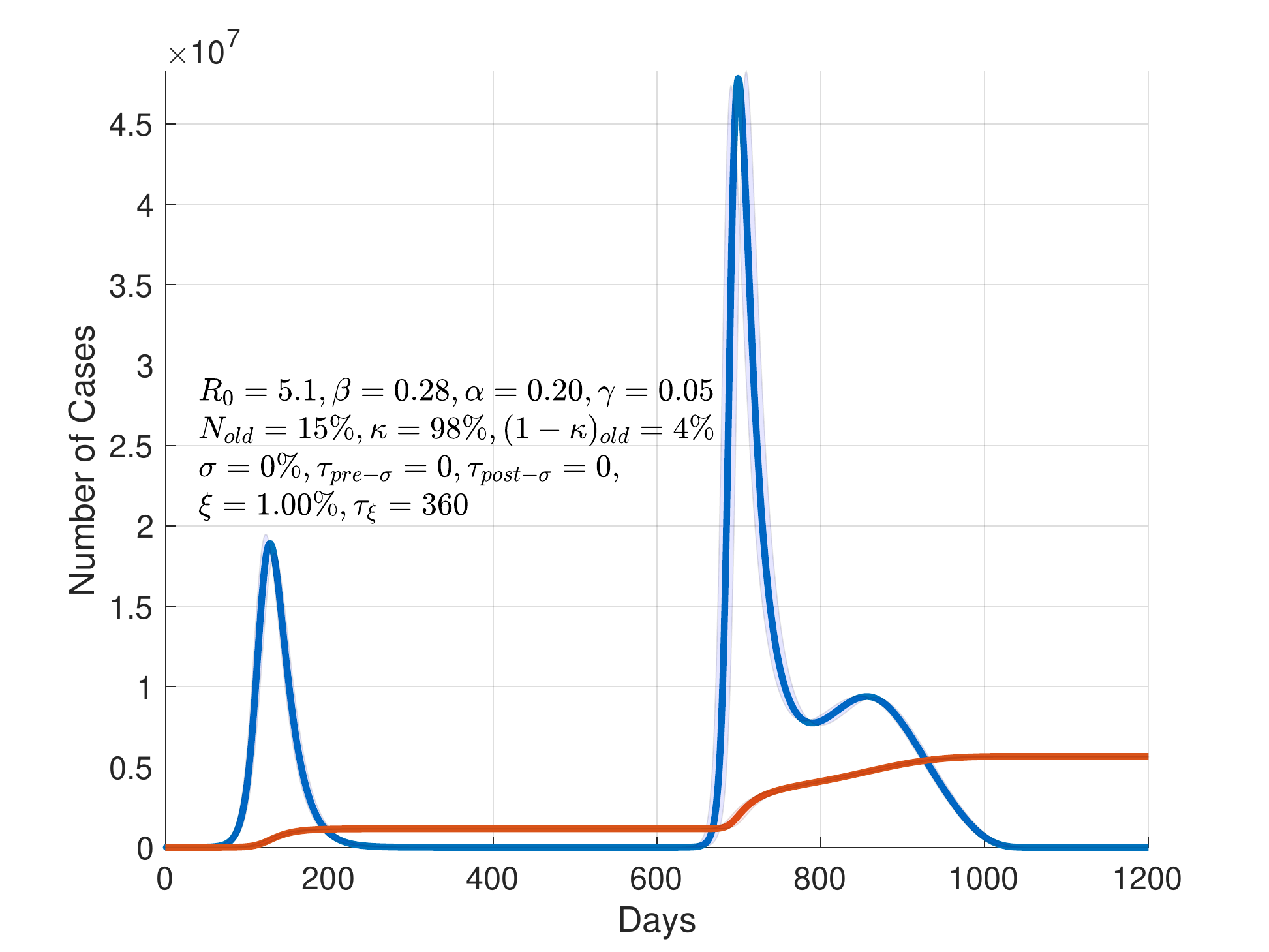}
              \caption{ }
              \label{fig:Korea7}
            \end{subfigure}
            \begin{subfigure}{\textwidth}
              \centering
              \includegraphics[width=0.3\textwidth]{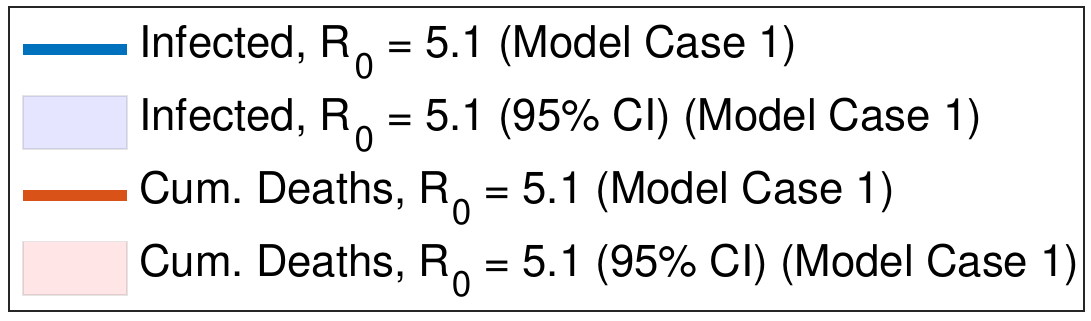}
            \end{subfigure}
            \caption{Trajectories for the infected and fatalities in South Korea due to resusceptibility where it is assumed that 1\% of the recovered cases are reinfected after a time span of temporal immunity of (a) 0 day, (b) 30 days, (c) 90 days, and (d) 360 days, respectively. However, these results only apply assuming that if there is no control action being taken to flatten the curve.}
            \label{fig:Korea3}
          \end{figure*}

          As such, we repeated the simulation for the Case Study on South Korea without control action, but with the inclusion of a possibility of resusceptibility. Here, we assumed that 1\% of the patients who recovered are resusceptible towards the virus ($\xi =$ 0.01), where the patients develop temporal immunity of $\tau_\xi = 0, 30, 90, 360$ days, respectively after recovering from the initial infection of the virus. Figure \ref{fig:Korea3} shows the simulation results, where the first infection spikes shown in all subfigures are synonymous to the result presented in Figure \ref{fig:KoreaNA2}. The subsequent infection spikes are the result of resusceptibility, depending on the days of temporal immune response. The results show new surges in infection cases after the specific $\tau_\xi$ in each case, which diminish over time as more people develop immunity towards the virus. Interestingly, for the result shown in Figure \ref{fig:Korea7} where $\tau_\xi = 360$ days, it could also be used to reflect on the situation where the virus may exhibit similar characteristics as the seasonal flu or the pandemic influenza A (pH1N1) that it is most likely active during certain seasons of the year, e.g. autumn/winter for the seasonal flu and spring/summer for the pH1N1, in which case an annual vaccine administration is necessary \cite{Kelly2009, Cook2011}.

          \subsection{Case Study 2: Prediction using Data in Northern Ireland}\label{NI}
          Given the location of which this research is based, data in Northern Ireland is used for prediction study of the model. The reports on confirmed and death cases are published daily since 24th March 2020 by the Northern Ireland Public Health Agency (PHA) via their Daily COVID-19 Surveillance Reports \cite{NI2020}. The first confirmed case was recorded on 27th February 2020, and as of 20th April 2020, the total number of confirmed cases stood at 2,728 with 207 fatalities.
          \begin{table}
            \centering
            \caption{Initial parameters used to fit the model to the data in Northern Ireland.}
            \label{tab:NI}
            {\small
            \begin{tabular}{lc}
              \toprule
              {\bf Parameter}                           & {\bf Value} \\
              \midrule
              Stock population, $N$                     & 1.88$\times 10^6$  \\
              Fraction of elderly population, $N_{old}$ & 0.18                \\
              Percentage of recovery, $\kappa$          & 0.94                \\
              Fatality rate for elderly, $1 - \kappa_{old}$ & 0.12            \\
              Incubation time, $\tau_{inc}$             & 5.1 days            \\
              Recovery time, $\tau_{rec}$               & 18.8 days           \\
              Basic reproduction number, $R_0$          & 5.0 (95\% CI: 4.85--5.15) \\
              Initial infected cases, $I(0)$            & 3                   \\
              Initial exposed cases, $E(0)$             & 60                  \\
              \bottomrule
            \end{tabular}}
          \end{table}

          We used the parameters in Table \ref{tab:NI} for the initial fitting of the model based on the data from PHA on the initial growth-rate of the epidemic in Northern Ireland. Figure \ref{fig:NINA1} shows the results of the initial fitting, with Figure \ref{fig:NINA2} depicting the projections of the infected and deaths if no control action is taken. We then simulated the model based on the control action carried out; most schools in Northern Ireland were closed beginning 18th March 2020 followed by a national lockdown initiated by the United Kingdom government on the 23rd March 2020. As such, we set $\tau_{pre-\sigma} = 20$ (20 days) to correspond to said dates since the first confirmed case, and assuming that it took a further approximately 12 days for the public to respond to these control action, i.e. $\tau_{post-\sigma} = 12$, we obtained the simulation results as shown in Figure \ref{fig:NIA}. The results show that in order for the model to follow the projected trajectories of the data from PHA in Figure \ref{fig:NI1}, the control action have to achieve an efficiency of about 45\% ($\sigma =$ 0.45), which indicates that the reproduction number could be reduced to $R_0 \approx$ 2.75. Comparing Figure \ref{fig:NINA2} with Figure \ref{fig:NI2}, the peak of the number of infected cases could be reduced from 650,000 cases to 350,000 cases.
          \begin{figure*}[t!]
            \centering
            \begin{subfigure}[t]{\columnwidth}
              \centering
              \includegraphics[width=\columnwidth]{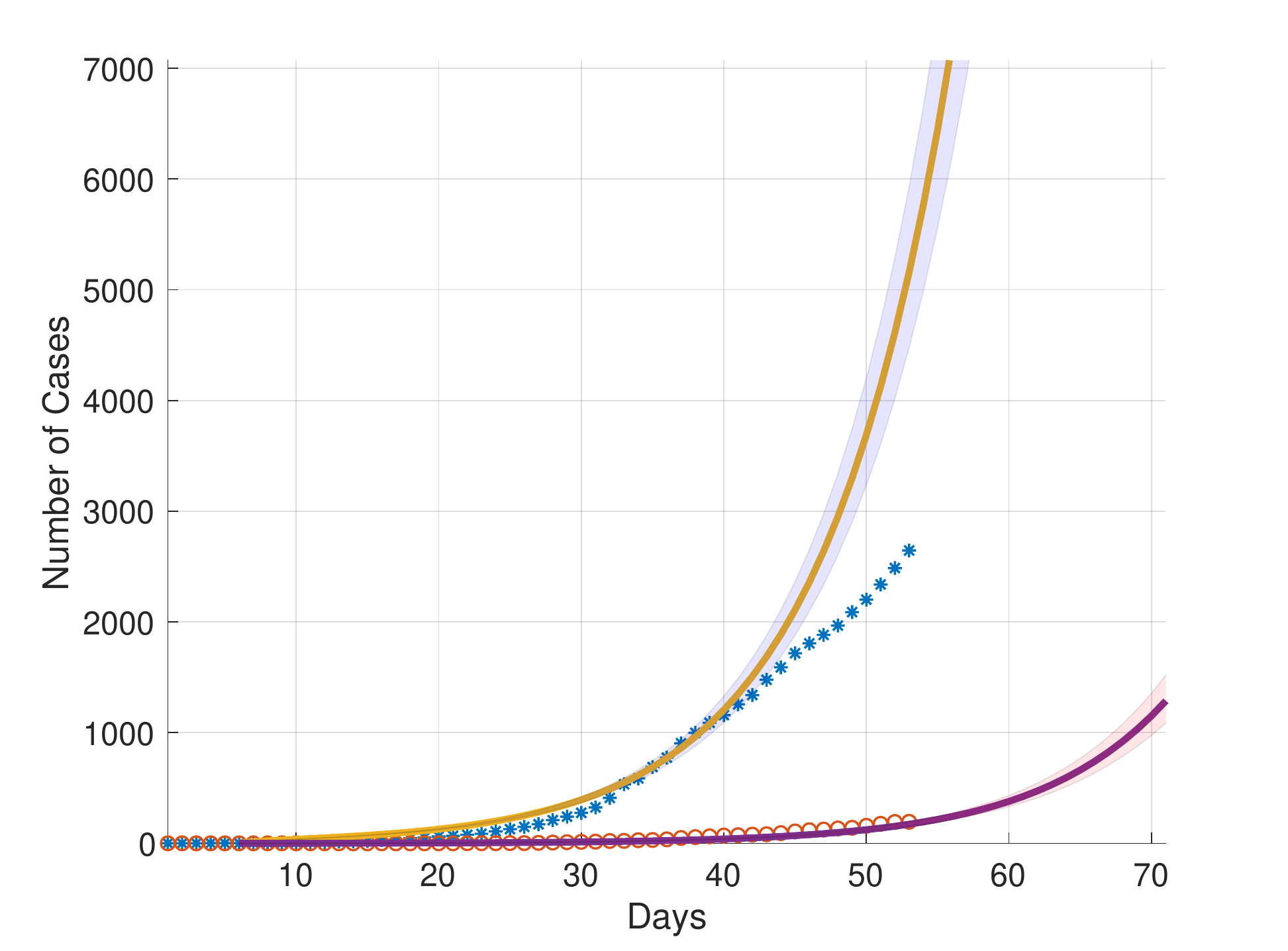}
              \caption{ }
              \label{fig:NINA1}
            \end{subfigure}
            \begin{subfigure}[t]{\columnwidth}
              \centering
              \includegraphics[width=\columnwidth]{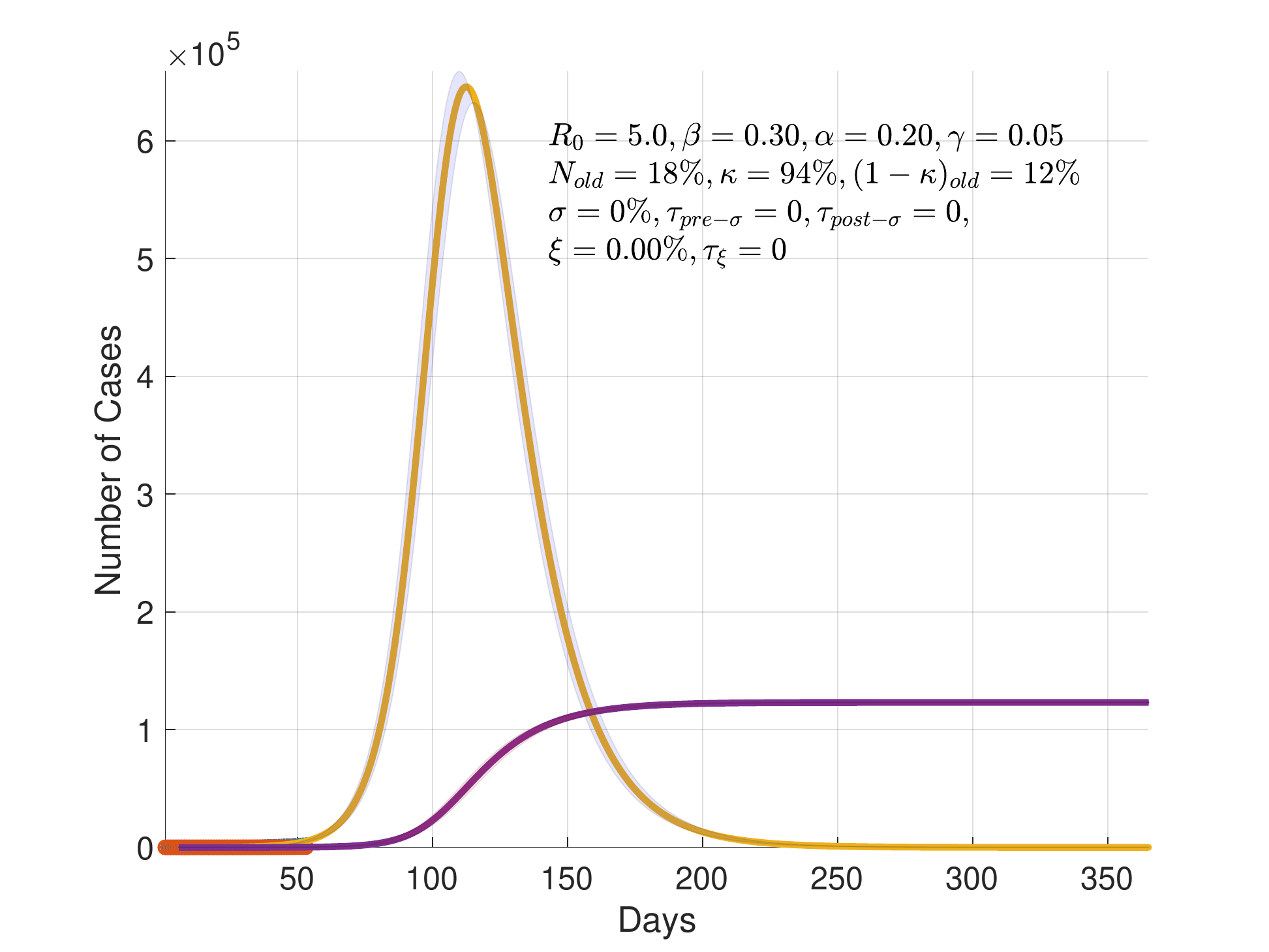}
              \caption{ }
              \label{fig:NINA2}
            \end{subfigure}
            \begin{subfigure}{\textwidth}
              \centering
              \includegraphics[width=0.3\textwidth]{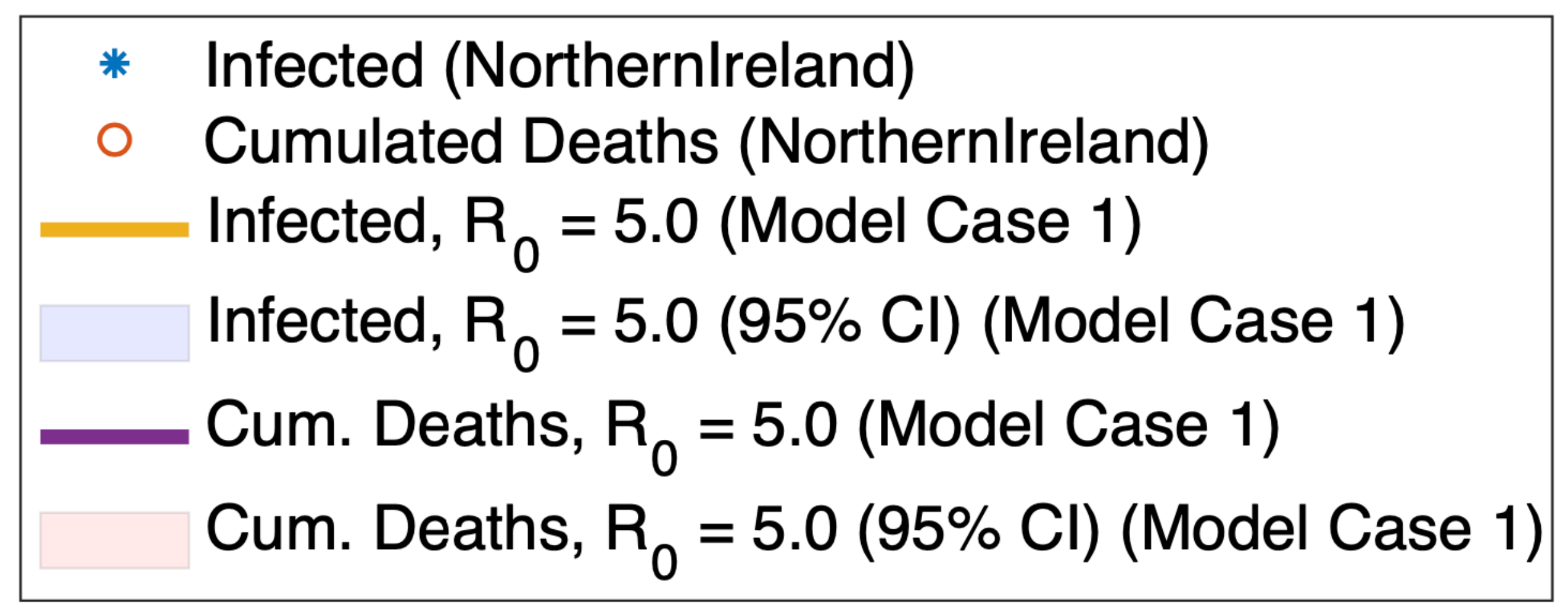}
            \end{subfigure}
            \caption{Subfigure (a) shows the initial fitting of the model onto the data in Northern Ireland and subfigure (b) shows the projections of the model when no control action is taken.}
            \label{fig:NINA}
          \end{figure*}

          \begin{figure*}[t!]
            \centering
            \begin{subfigure}{\columnwidth}
              \centering
              \includegraphics[width=\columnwidth]{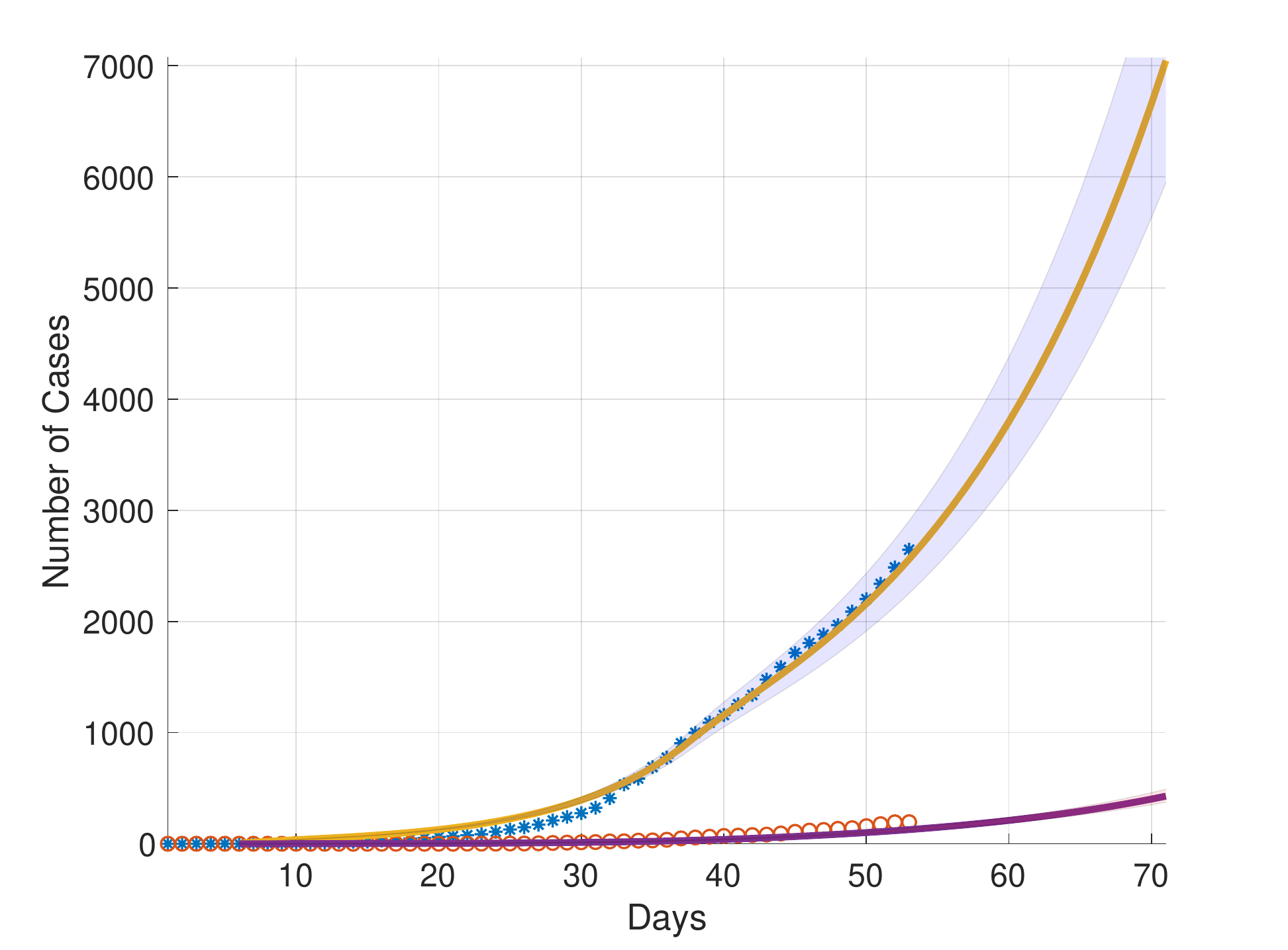}
              \caption{ }
              \label{fig:NI1}
            \end{subfigure}
            \begin{subfigure}{\columnwidth}
              \centering
              \includegraphics[width=\columnwidth]{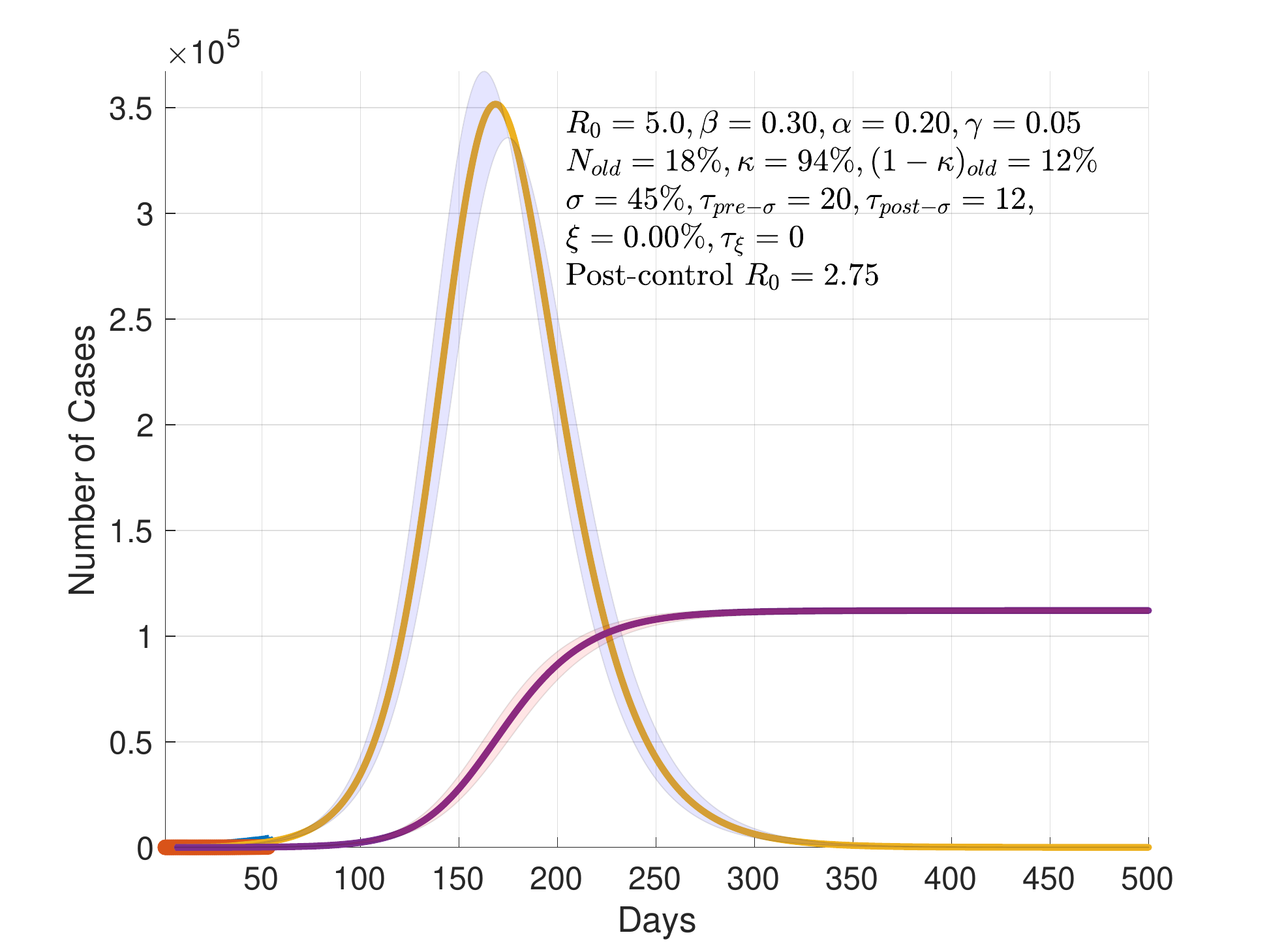}
              \caption{ }
              \label{fig:NI2}
            \end{subfigure}
            \begin{subfigure}{\textwidth}
              \centering
              \includegraphics[width=0.3\textwidth]{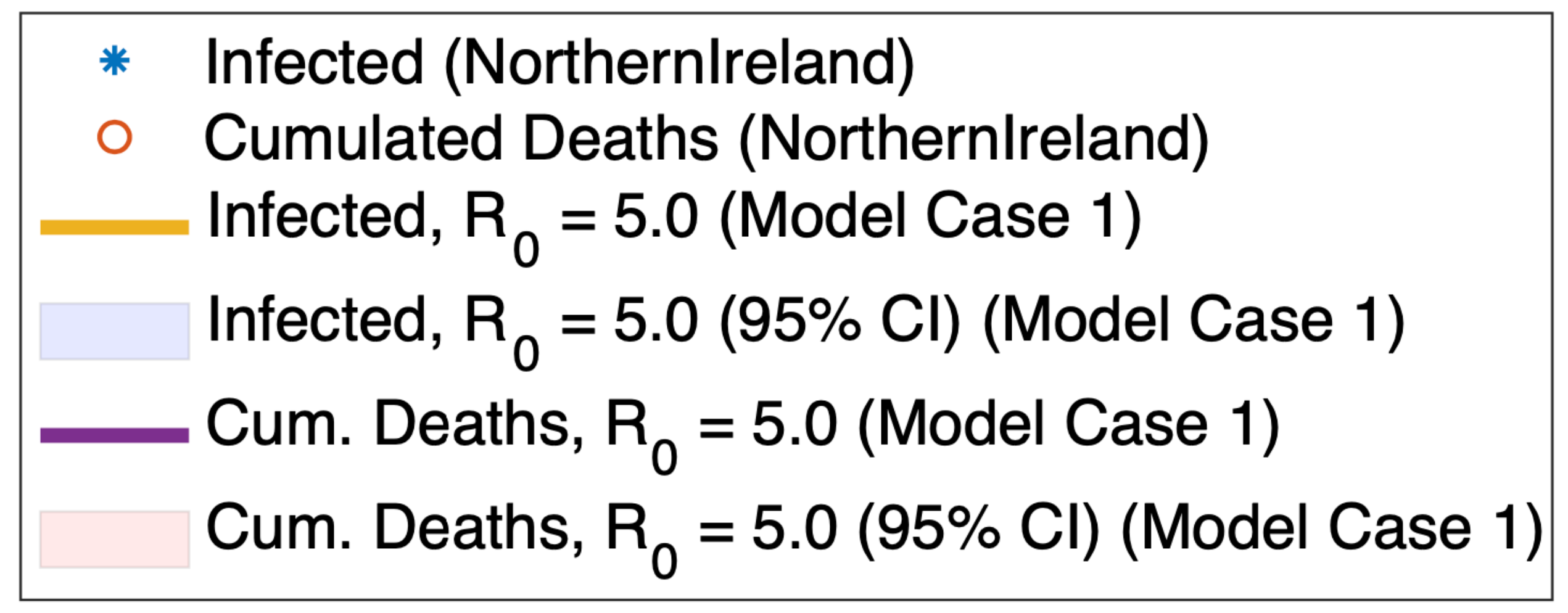}
            \end{subfigure}
            \caption{The projections of the model onto the data in Northern Ireland when control action with an efficiency of 45\% is taken. Subfigure (a) shows the projections during the first 70 days while subfigure (b) shows the projections until the system achieves stability.}
            \label{fig:NIA}
          \end{figure*}

          \subsubsection{Further Control Action To Meet Critical Care Capacity}
          However, based on the results shown in Figure \ref{fig:NIA}, it is essential to further flatten this curve due to the limit of about 330 critical care beds available in Northern Ireland (100 initial setup + 230 introduced by the newly built Nightingale hospital) \cite{BBC2020}. According to the Intensive Care National Audit \& Research Centre (ICNARC) with its ``Report on 2249 patients critically ill with COVID-19'' dated 4th April 2020, about 6\% of those tested positive for the SARS-CoV-2 required critical care \cite{ICNARC2020}. Meanwhile in Italy, as of 29th March 2020, up to 12\% of all positive cases were admitted to the intensive care unit (ICU) \cite{Phua2020}. As such, should we assume that approximately 10\% of those tested positive in Northern Ireland would require ICU admission, then the peak of the infection curve should not exceed 3,300 cases, i.e. more control action have to be taken to reduce the peak of 350,000 cases as seen in Figure \ref{fig:NI2}.

          Therefore, on day 38, which is about one week after the infection curve started to flatten due to the first control action, a second control action was introduced into the model. This second control action also reflects on the announcement made by the United Kingdom government in early April 2020 to allow police officers to enforce social distancing measures. Assuming that this second control action results in a further efficiency of 66\%, the reproduction number could then be reduced to $R_0 \approx$ 0.93, and that it would take another 7 days for the public to fully respond to the control action, we could observe the results as shown in Figure \ref{fig:NISA}. With the initiation of the second control action, it can be seen in Figure \ref{fig:NISA1} that the peak in the infection curve is now reduced to 3,500 cases. As such, the critical care capacity should be able to meet the demand for treatment based on the same assumption made earlier in this section, where it is estimated about 10\% of the infected cases are admitted to the ICU. Another observation worth noting is that the number of deaths in the worst case scenario has also been reduced to about 3,000 cases. See Figure \ref{fig:NISA2}. This value echos the projection made by the government in Northern Ireland that COVID-19 could lead to 3,000 deaths during the first wave \cite{BBCNI2020}.
          \begin{figure*}[t!]
            \centering
            \begin{subfigure}{\columnwidth}
              \centering
              \includegraphics[width=\columnwidth]{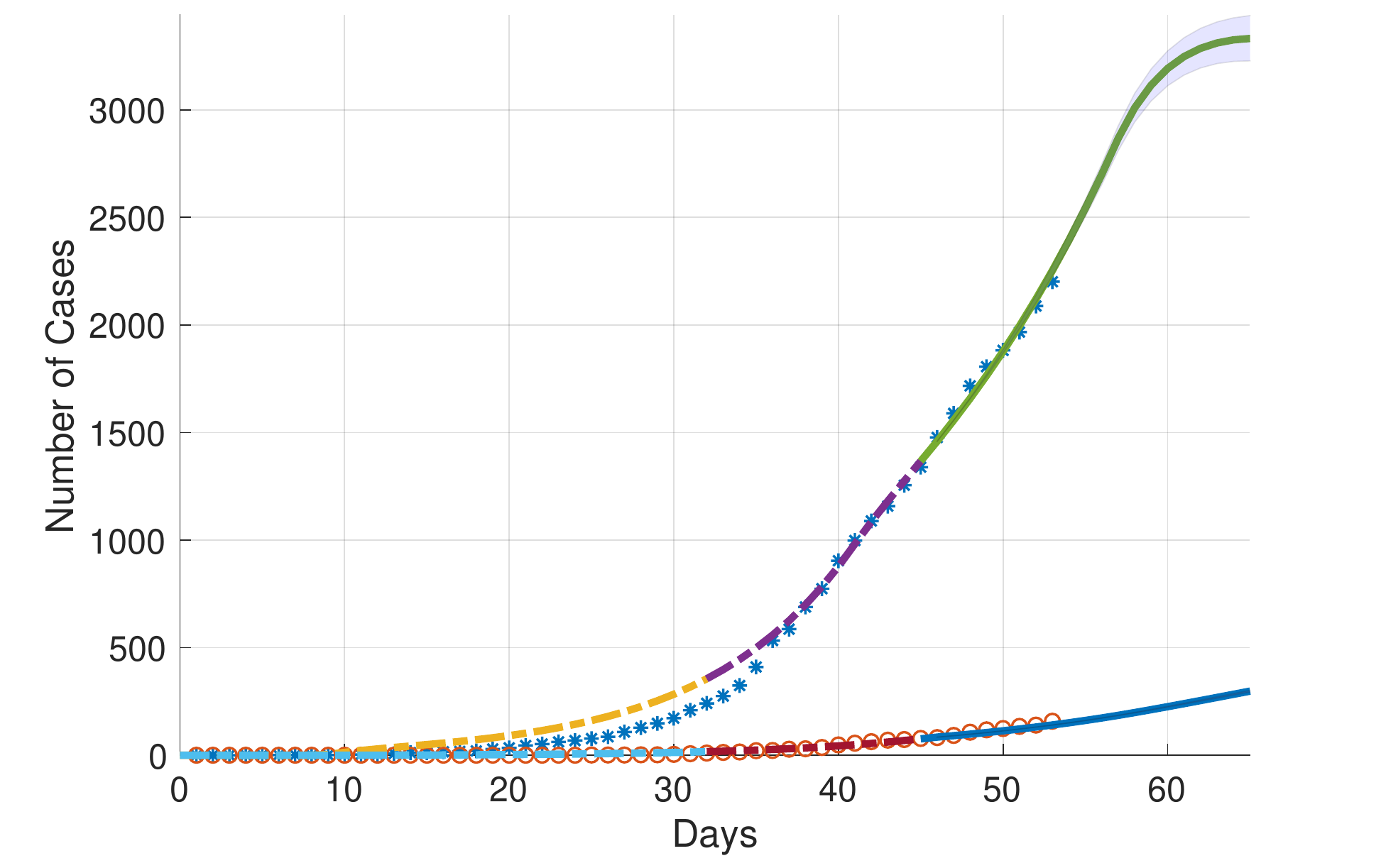}
              \caption{ }
              \label{fig:NISA1}
            \end{subfigure}
            \begin{subfigure}{\columnwidth}
              \centering
              \includegraphics[width=\columnwidth]{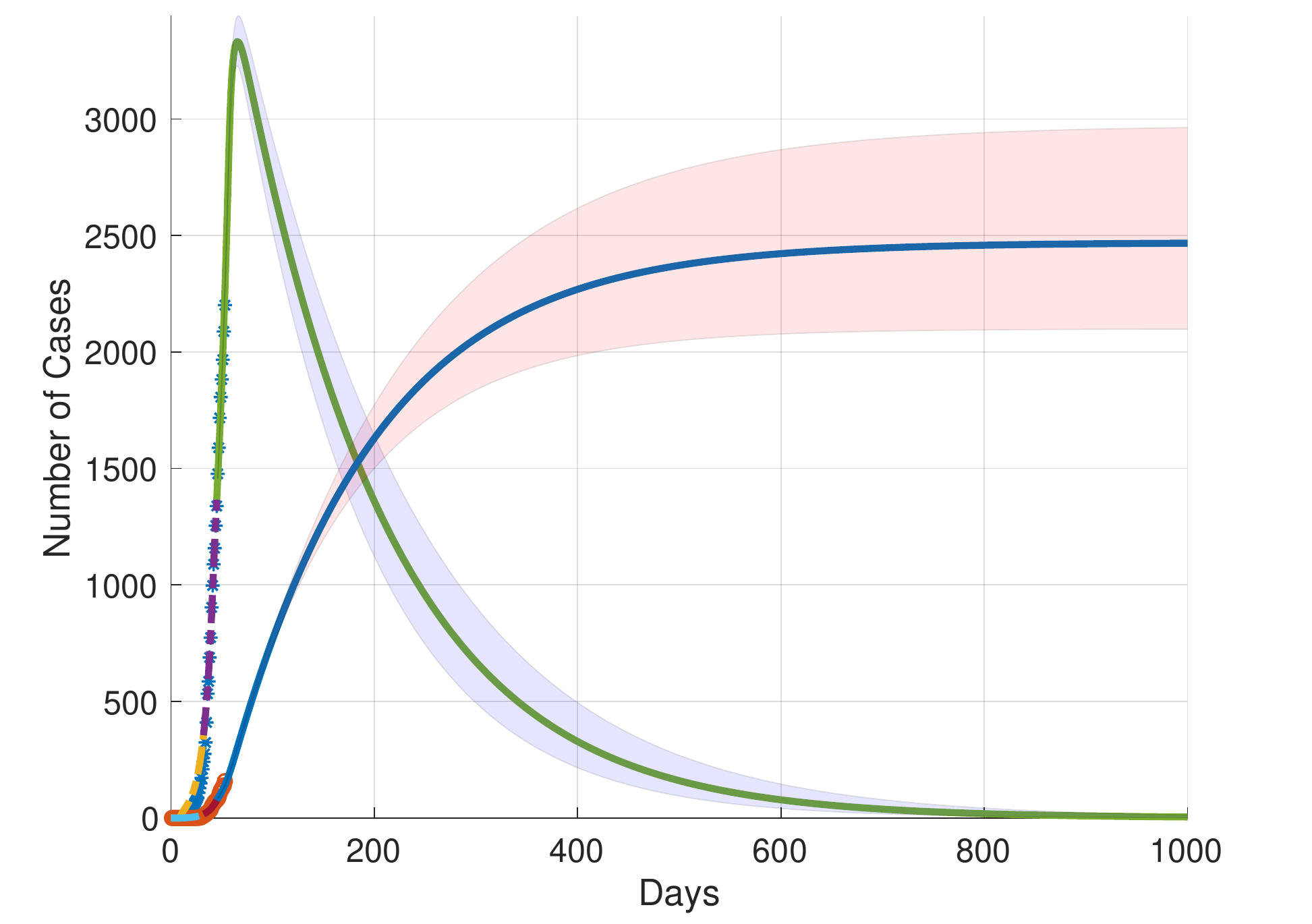}
              \caption{ }
              \label{fig:NISA2}
            \end{subfigure}
            \begin{subfigure}{\textwidth}
              \centering
              \includegraphics[width=0.3\textwidth]{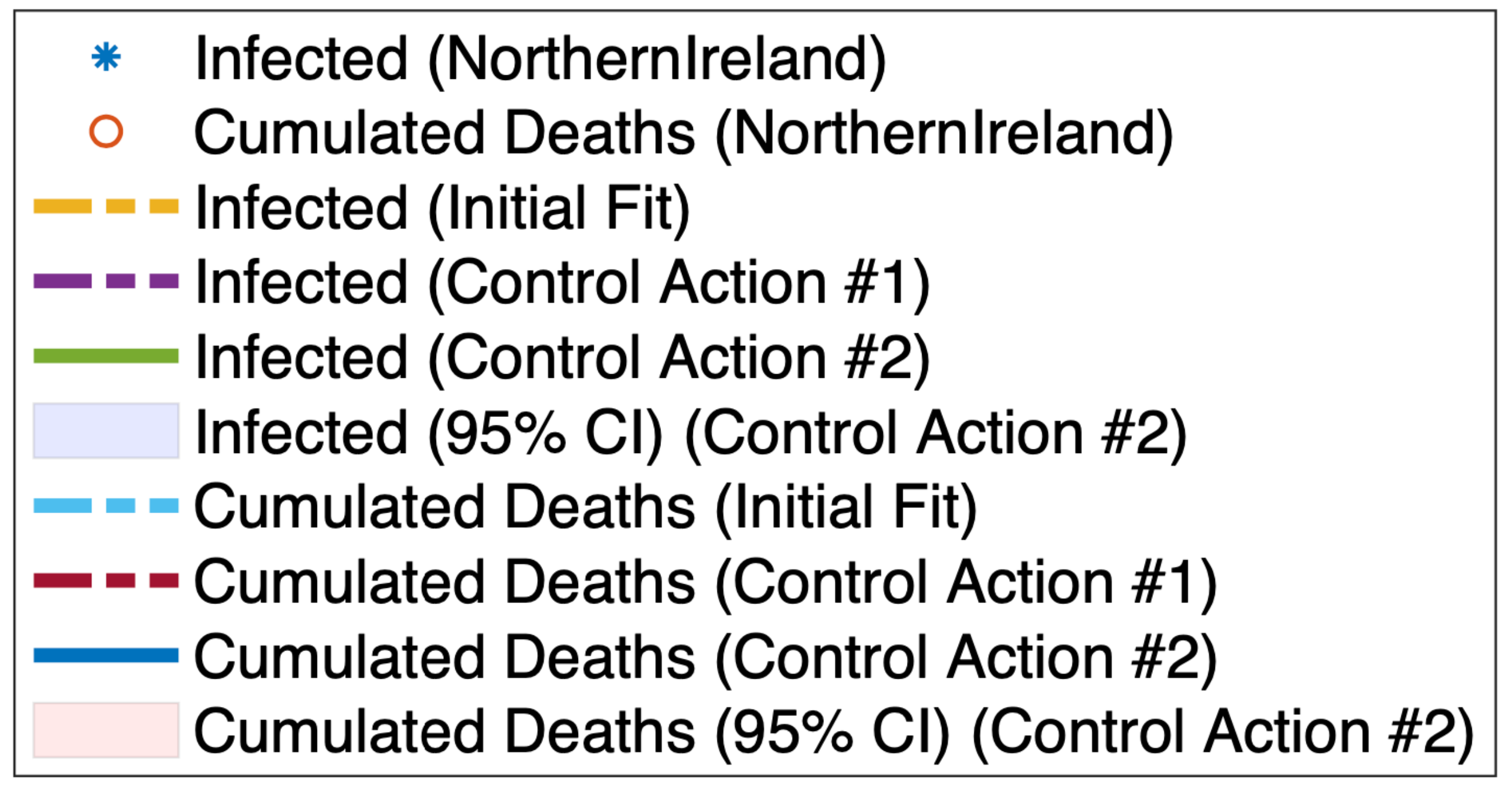}
            \end{subfigure}
            \caption{Subfigure (a) shows the initial projections of the infected and cumulative deaths curves in Northern Ireland after the second control action with an efficiency of 66\% was introduced on day 38 of the model. Subfigure (b) shows the remaining projections of the model.}
            \label{fig:NISA}
          \end{figure*}

          \section{Conclusion}\label{Conclusion}
          {This paper has presented a robust model for COVID-19 based on a modified SEIRS method to include considerations for the ageing population, and time delay for control action as well as resusceptibility of the recovered population due to temporal immunity. Two case studies using real-world data were presented in this research; the first case for verification of the model based on the data in South Korea including a study on the possibility of resusceptibility of recovered population; and the second case for prediction study of the model using data and up-to-date control action and related events in Northern Ireland. The simulation results from the case studies have clearly shown that the time of which the control action is taken and also the time for the public to properly respond to such intervention measures are critical in helping to flatten the curve. Also, until a time where a vaccine is developed and made available to the general public, the possibility of resusceptibility, no matter how small, will lead to subsequent waves of infections in the future depending on the time of temporal immunity. A simulation package was developed using the MATLAB/Simulink environment to ease understanding on the spread of the virus as well as the efficiency that needs to be achieved by the control action in order to successfully flatten the infection curve to not overload the healthcare capacity.

          Interesting future research and expansion of the model include but not limited to the predictions for the occupancy of ICU beds, the effects of easing control action on $R_0$ and hence, the time of which control action have to be reinstigated, as well as specific subregions analysis such as demographic information to model the transmission of the virus among subregions in the country to cater for population movements and travels.


          \bibliographystyle{elsarticle-num}
          \bibliography{refs}

          %
          %
          %
          %

          \appendix
          \section{Description of The Simulation Package}\label{Simulation}
          Figure \ref{fig:GUI} shows the graphical user interface (GUI) of the simulation package developed using the MATLAB/Simulink environment. Users can use this interface to set preferred settings for the simulation and also to view simulation results. The simulation kit can be downloaded from \url{https://github.com/nkymark/COVIDSim}.
          \begin{figure*}[t!]
            \centering
            \includegraphics[width=\textwidth]{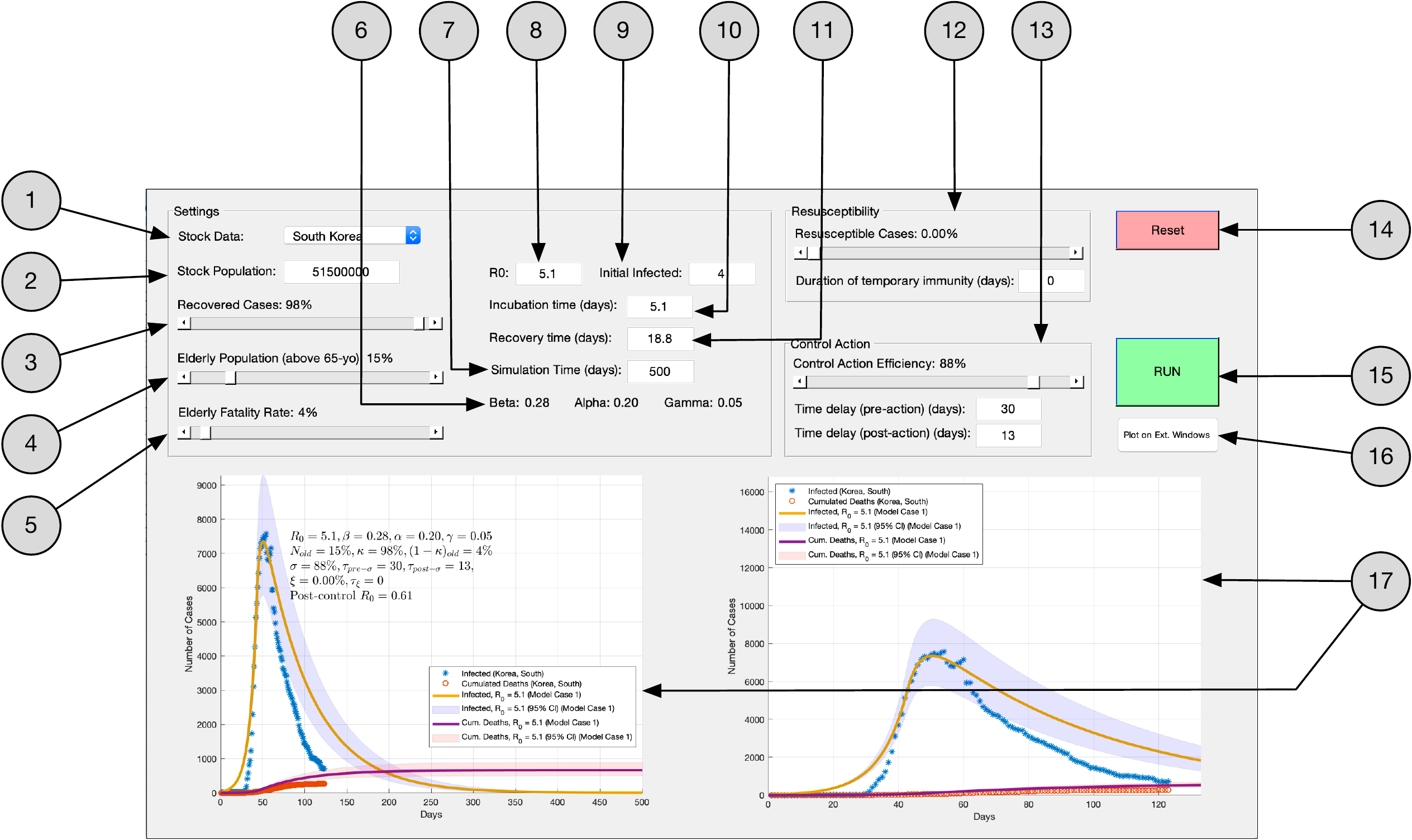}
            \caption{The main graphical user interface of the simulation package in MATLAB. \textcircled{\tiny 1} Load real-world data for the selected country. \textcircled{\tiny 2} Set the stock population $N$ for simulation. \textcircled{\tiny 3} Set the percentage of recovered cases $\kappa$. \textcircled{\tiny 4} Set the fraction of elderly population $N_{old}$. \textcircled{\tiny 5} Set the fatality rate for the elderly population $(1-\kappa_{old})$. \textcircled{\tiny 6} Computed values for $\beta = R_0 \gamma,~ \alpha = \frac{1}{\tau_{inc}}$, and $\gamma = \frac{1}{\tau_{rec}}$ using values entered for $R_0, \tau_{inc},$ and $\tau_{rec}$. \textcircled{\tiny 7} Set the simulation time in days. \textcircled{\tiny 8} Set the value for the basic reproduction number $R_0$. \textcircled{\tiny 9} Set the initial number of infected cases $I(0)$. \textcircled{\tiny 10} Set the incubation time $\tau_{inc}$. \textcircled{\tiny 11} Set the recovery time $\tau_{rec}$. \textcircled{\tiny 12} Settings for recusceptibility, including the percentage of resusceptible cases $\xi$ and duration of temporal immunity $\tau_\xi$. \textcircled{\tiny 13} Settings for control action, including the effeciency rate $\sigma$ as well as the time delay during pre- and post-control action, $\tau_{pre-\sigma}$ and $\tau_{post-\sigma}$, respectively. \textcircled{\tiny 14} Reset the GUI and clear all plots. \textcircled{\tiny 15} Run the simulation. \textcircled{\tiny 16} Recreate the graphs on external Matlab figure windows. \textcircled{\tiny 17} Graphical plots from the simulation (left figure for overall simulation while right figure compare initial projections of the model with real-world data).}
            \label{fig:GUI}
          \end{figure*}

          \subsection{Establishing Simulation Settings}\label{Settings}
          At the top section of the GUI are some interactive interfaces available for the user to set key simulation settings, which include the following:
          \begin{itemize}
            \item General Settings:
            \begin{itemize}
              \item {\it Stock Data}: Use this to load real-world data of select countries. The data are obtained from \cite{LancetDong}.
              \item {\it Stock Population}: The stock population $N$ is entered here.
              \item {\it Recovered Cases}: Use this to set the percentage of recovered cases $\kappa$.
              \item {\it Elderly Population}: Use this to set the fraction of elderly population (above 65 years of age) $N_{old}$.
              \item {\it Elderly Fatality Rate}: Use this to set the fatality rate $(1-\kappa_{old})$ for the elderly population.
              \item {\it SEIR Parameters}: Use this to set the values for $R_0, \tau_{inc}, \tau_{rec}$, the initial infected cases $I(0)$, and the simulation time.
            \end{itemize}
            \item Resusceptibility Settings:
            \begin{itemize}
              \item {\it Resusceptible Cases}: Use this to set the percentage of recovered cases who are resusceptible.
              \item {\it Duration of temporal immunity}: Use this to set the time of short-term immune response $\tau_\xi$, assuming there is no permanent immunity after recovery.
            \end{itemize}
            \item Control Action Settings:
            \begin{itemize}
              \item {\it Control Action Efficiency}: Use this to set the percentage of control action efficiency $\sigma$.
              \item {\it Pre-action Time Delay}: Use this to set the time delay $\tau_{pre-\sigma}$ for the control action to be introduced after the first confirmed case.
              \item {\it Post-action Time Delay}: Use this to set the time delay $\tau_{post-\sigma}$ to mimic the time it takes for the population to respond to the control action.
            \end{itemize}
          \end{itemize}

          \subsection{Simulation Results}\label{Results}
          The simulation results are displayed at the bottom section of the GUI. The plot on the right shows the initial fit of the model using the settings established in Section \ref{Settings} onto the real-world data of the select country, while the plot on the left shows the simulation results until the simulation stop time.

          \end{document}